\documentclass[10pt]{amsart}

\usepackage[english]{babel}
\usepackage[left=3cm,right=3cm,top=2.5cm,bottom=3cm]{geometry}

\usepackage{amsmath,amssymb,amsthm,bbm,color,graphics,version, mathtools,cases}
\usepackage{mathrsfs}
\usepackage{amsfonts}
\usepackage{graphicx}
\usepackage{setspace}
\usepackage[T1]{fontenc}
\usepackage[english]{babel}
\usepackage[utf8]{inputenc}
\usepackage{mathtools}
\usepackage{dsfont}
\usepackage[shortlabels]{enumitem}
\usepackage{bbm}
\usepackage{url}
\usepackage[dvipsnames]{xcolor}
\usepackage[colorinlistoftodos]{todonotes}

\newtheorem{theorem}{Theorem}
\newtheorem{assumption}{Assumption}

\newtheorem{condition}[theorem]{Condition}
\newtheorem{conjecture}[theorem]{Conjecture}
\newtheorem{corollary}[theorem]{Corollary}

\newtheorem{definition}[theorem]{Definition}
\newtheorem{example}[theorem]{Example}

\newtheorem{lemma}[theorem]{Lemma}

\newtheorem{proposition}[theorem]{Proposition}
\newtheorem{remark}[theorem]{Remark}

\newcommand{\ubar}[1]{\text{\b{$#1$}}}

\title{Perpetual American options with asset-dependent discounting}

\author{Jonas Al-Hadad}
\address{Faculty of Pure and Applied Mathematics\\
Wrocław University of Science and Technology\\
ul. Wyb. Wyspiańskiego 27, 50-370 Wrocław \\
Poland}
\email{jonas.al-hadad@pwr.edu.pl}

\author{Zbigniew Palmowski}
\address{Faculty of Pure and Applied Mathematics\\
Wrocław University of Science and Technology\\
ul. Wyb. Wyspiańskiego 27, 50-370 Wrocław \\ Poland}
\email{zbigniew.palmowski@pwr.edu.pl}

\thanks{Jonas Al-Hadad and Zbigniew Palmowski have been partially supported by the National Science Centre under the grant 2016/23/B/HS4/00566.}

\date{\today}

\begin{document}

\begin{abstract}
In this paper we consider the following optimal stopping problem
$$V^{\omega}_{\text{\rm A}}(s) = \sup_{\tau\in\mathcal{T}}
\mathbb{E}_{s}[e^{-\int_0^\tau \omega(S_w) dw} g(S_\tau)],$$
where the process $S_t$ is a jump-diffusion process, $\mathcal{T}$ is a family of stopping times while $g$ and $\omega$ are fixed payoff function and discount function, respectively.
In a financial market context, if $g(s)=(K-s)^+$ or $g(s)=(s-K)^+$ and $\mathbb{E}$ is the expectation taken with respect to a martingale measure,
$V^{\omega}_{\text{\rm A}}(s)$ describes the price of a perpetual American option with a discount rate depending on the value of the asset process $S_t$.
If $\omega$ is a constant, the above problem produces the standard case of pricing perpetual American options.
In the first part of this paper we find sufficient conditions for the convexity of the value function $V^{\omega}_{\text{\rm A}}(s)$.
This allows us to determine the stopping region as a certain interval and hence we are able to identify the form of $V^{\omega}_{\text{\rm A}}(s)$.
We also prove a put-call symmetry for American options with asset-dependent discounting.
In the case when $S_t$ is a geometric L\'evy process we give exact expressions using the so-called omega scale functions introduced in \cite{LiBo}. We prove that the analysed value function satisfies the HJB equation and we give sufficient conditions for the smooth fit property as well.
Finally, we present a few examples for which we obtain the analytical form of the value function $V^{\omega}_{\text{\rm A}}(s)$.

\vspace{3mm}


\end{abstract}

\maketitle

\pagestyle{myheadings} \markboth{\sc J.\ Al-Hadad --- Z.\ Palmowski
} {\sc Perpetual American options with asset-dependent discounting}

\vspace{1.8cm}

\tableofcontents

\newpage

\section{Introduction}
In this paper
the uncertainty associated with the stock price $S_t$ is described by a jump-diffusion process
defined on a complete filtered risk-neutral probability space
$(\Omega, \mathcal{F}, \{\mathcal{F}_t\}_{t\geq 0}, \mathbb{P})$, where
$\mathcal{F}_t$ is a natural filtration of $S_t$ satisfying the usual conditions and $\mathbb{P}$ is a risk-neutral measure under which the discounted (with respect to a risk-free interest rate) asset price process $S_t$ is a local martingale.
We point out that, as noted in  \cite{cont}, introducing jumps into the model, implies lost of completeness of the market
which results in the lack of uniqueness of the equivalent martingale measure.
Our main goal is the analysis of the following optimal stopping problem
\begin{equation}\label{mainProblem}
V^{\omega}_{\text{\rm A}}(s) := \sup_{\tau\in\mathcal{T}} \mathbb{E}_{s}\left[e^{-\int_0^\tau \omega(S_w) dw} g(S_\tau)\right],
\end{equation}
where $\mathcal{T}$ is a family of $\mathcal{F}_t$-stopping times while $g$ and $\omega$ are fixed payoff function and discount function, respectively.
Above $\mathbb{E}_s$ denotes the expectation with respect to $\mathbb{P}$ when $S_0=s$.
We assume that the function $g$ is convex.
We allow in this paper for $\omega$ to have negative values as well.
In the case when $g(s)=(K-s)^+$ or $g(s)=(s-K)^+$ and $\mathbb{P}$ is a {\bf martingale measure},
this function can be interpreted as the value function of a {\bf perpetual American option with asset-dependent discounting
\footnote{Throughout the paper, we use the terms \textit{asset-dependent discounting} and \textit{functional discounting} interchangeably.}} $\omega$ and payoff function $g$.
In the case of the general theory of stochastic processes, multiplying by the discount factor $e^{-\int_0^\tau \omega(S_w) dw}$
corresponds to killing of a generator of $S_t$ by potential $\omega$.

This problem extends the classical theory of option pricing, where the deterministic discount rate is considered, that is if $\omega(s)=r$, then
we obtain the standard form
\begin{equation*}\label{standardProblem}
V_{\text{\rm A}}(s) := \sup_{\tau\in\mathcal{T}}\mathbb{E}_{s}\left[e^{-r\tau} g(S_\tau)\right]
\end{equation*}
of the perpetual American option's value function with constant discount rate $r$.

The main objective of this paper is to {\bf find a closed expression of \eqref{mainProblem} and identify the optimal stopping rule $\tau^*$} for which the supremum is attained.
To do this, we start from proving in Theorem \ref{mainTheoremAmericanOptionConvexity} an {\bf inheritance of convexity property} from the payoff function to the value function.
This corresponds to preserving the convexity by the solution to a certain obstacle problem.

Using this observation and the classical optimal stopping theory presented e.g. in \cite{peskir} one can identify the {\bf optimal stopping region as an interval} $[l^*,u^*]$, that is, $\tau^*=\inf\{t\geq 0: S_t\in [l^*,u^*]\}$.
Hence, in general, one can obtain in this case a {\it double continuation region}.

Later we focus on the case when $S_t$ is a geometric spectrally negative L\'evy process, that is,
$$S_t:=e^{X_t}$$ for a spectrally
negative L\'{e}vy process $X_t$.
In this case, using the fluctuation theory of L\'evy processes, {\bf we identify the value function \eqref{mainProblem} in terms of the omega scale functions} introduced in \cite{LiBo}.

For optimal stopping problem \eqref{mainProblem} {\bf we give sufficient conditions under which
we can formalise the classical approach} here as well.
In particular, in Theorem \ref{smoothTheorem}
we prove that if the value function $V^{\omega}_{\text{\rm A}}(s)$ is smooth enough then
it is the {\bf unique solution to a certain Hamiltonian-Jacobi-Bellman (HJB) system}.
Moreover, in the case of geometric L\'evy process of the asset price $S_t$, we prove that the regularity of $1$ for $(0, 1)$ and $(1,+\infty$)
gives the {\bf smooth fit property at the ends of the stopping region}.
We want to underline here that proving the above mentioned regularity of the value function
in the case of the jump-diffusion processes (which allows to
formulate the HJB equation) in general is very difficult task (see \cite{peskirTiziano} for some deep results
related with it). Nevertheless, it is possible in our case thanks
to Theorem \ref{smoothTheorem} and Remark \ref{smoothvalue}. Having the convexity of the value function proved and hence having the value
function defined in terms of the omega scale functions, {\bf we can prove the appropriate smoothness condition
using the fluctuation theory of L\'evy processes.
Only then one can apply the HJB equation.}

Further, even solving the HJB equation does not give in the straightforward way the form of the
the stopping region (except, of course, the fact that it is the set where the value function equals
the payoff function). This is the reason why we do not follow this path, but rely on our approach allowing us to define precisely the stopping region.

In Theorem \ref{putCallSymmetry} we show the put-call symmetry which holds in our setting as well.

These theoretical results allow us to find the {\bf price of the perpetual American option with asset-dependent discounting for some particular cases.}
We take for example a put option, that is $g(s)=(K-s)^{+}$, and
a geometric Brownian motion for the asset price $S_t$. We model $S_t$ also by the geometric L\'evy process with exponentially distributed
downward jumps.
We analyse various discount functions $\omega$.
In Section \ref{sec:examples} we provide a few examples for which we obtain the analytical form of the value function.

The discount rate changing in time or a random discount rate
are widely used in pricing derivatives in financial markets. They have proved to be
valuable and flexible tools to identify the value of various options.
Usually, either the interest rate is independent from the asset price or this dependence is introduced via
taking a correlation between gaussian components of these two processes. Our aim is completely different.
We want to {\bf understand} an extreme case when we have {\bf strong, functional dependence between the interest rate and the asset price.}
In particular, we take a close look at the American put option with the {\bf discount function $\omega$ having the opposite monotonicity to the payoff function}. At first sight, such a case seems to be counter-intuitive, because, for the put option, if the asset price is in higher region one can expect that the interest rate will be lower and the opposite effect one expects for smaller range of asset's prices. This dependence somehow balances the discounting function with the payoff function.
However, we can think of an investor who have a strong confidence in the movement of the asset price and wishes to make an extra profit when he is right and suffers a greater loss when he is wrong. This concept resembles an idea that stands behind barrier options, i.e. if an investor believes that it is unlikely that the asset price will hit a given level, he can add a knock-out provision with the barrier set at the support level, so he can reduce the price of an option. By including the barrier provision, he can eliminate paying for those scenarios he feels are unlikely.
In our approach, we work in two ways by reducing the premium thanks to incidents being improbable from the investor's perspective and increasing it for scenarios that are more likely to happen.
Such a description of the analysed option adequately describes a financial instrument tailored to the risky investor.
Let us add that for example up-and-out put options analysed by \cite{Linetsky}
is a particular case of our option.

One can look at optimisation problem \eqref{mainProblem}
from a wider perspective though. The killing by potential $\omega$ has been known widely in physics and other applied sciences.
Then \eqref{mainProblem} can be seen as a certain functional describing gain or energy and the goal is to optimise it by choosing
some random stopping time. We focus here on financial applications only and therefore {\bf throughout this paper
the price is calculated under the martingale measure.}
It means more formally that there exists a risk-free interest rate $r$ such that the discounted price process is a local martingale under $\mathbb{P}$.
Note that the discount function $\omega$ may equal the risk-free interest rate $r$ but in our case usually is different.

For example, in the case of a {\bf gold loan}
(see \cite{CS} for the survey related with this financial instrument)
a borrower receives at time $0$ (the date of contract inception) a loan amount $K>0$ using
one mass unit (one troy ounce, say) of gold as collateral, which must be physically delivered to a lender.
This amount grows at the functional borrowing rate
given in the contract that can depend on the gold spot price $\bar{S}_t$.
When paying back the
loan, the borrower can redeem the gold at any time and then the contract is terminated.
Of course, the dynamic of $\bar{S}_t$ under the risk-neutral measure $\mathbb{P}$ is such that
the discounted price $e^{-rt}\bar{S}_t$ is a martingale, that is
$\mathbb{E} \bar{S}_t = e^{rt}\mathbb{E} \bar{S}_0$.
Assume that the costs of storing
equals the borrowing rate plus some fixed costs $c>0$ per unit of time
and the borrowing rate is a function $\bar{\omega}$ of the gold spot price $\bar{S}_te^{ct}$
discounted by this fixed costs. Then
the value of the contract, with infinite maturity date,  at time $0$ equals
\begin{equation*}\label{pricegold}
\sup_{\tau \in \mathcal{T}} \mathbb{E}\left[\left.e^{-r\tau}\left(\bar{S}_\tau e^{\int_0^\tau \bar{\omega} (\bar{S}_we^{cw})dw +c\tau}-Ke^{\int_0^\tau \bar{\omega} (\bar{S}_we^{cw})dw}\right)^+\right|\ \bar{S}_0=s\right]=
\sup_{\tau \in \mathcal{T}} \mathbb{E}_s\left[e^{-\int_0^\tau \omega (S_w)dw}\left(S_\tau-K\right)^+\right]\end{equation*}
where $S_t=\bar{S}_te^{ct}$, $\omega (S_t)=r-\bar{\omega}(S_t)$ and
$\mathbb{E}_s S_t= e^{ct}\mathbb{E}_s \bar{S}_t=e^{(r+c)t}\mathbb{E}S_0$.
%
%

{\bf Our research methodology is based on combining the theory of partial differential equations with the fluctuation theory of L\'evy processes.}

To prove the convexity we start from proving in Theorem \ref{TheoremAboutConvexityOfValueFunctionDiffusionWithJumps}
the convexity of
\begin{equation}\label{valueFunctionOmegaEuropeanOption}
V^{\omega}_{\text{\rm E}}(s, t) := \mathbb{E}_{s, t}\left[e^{-\int_{t}^{T}\omega(S_w) d w} g(S_T)\right]
\end{equation}
for fixed time horizon $T$, where $\mathbb{E}_{s,t}$ is the expectation $\mathbb{E}$ with respect to $\mathbb{P}$ when
$S_t=s$.
In the proof we follow the idea given by Ekstr\"{o}m and Tysk in \cite{propertiesOfOptionPrices}.
Namely, the value function $V^{\omega}_{\text{\rm E}}(s,t)$ given in \eqref{valueFunctionOmegaEuropeanOption} can be presented as the unique viscosity solution to a certain
Cauchy problem for some second-order operator related to the generator of the process $S_t$.
In fact, applying similar arguments like in \cite[Proposition 5.3]{pham98} and \cite[Lemma 3.1]{propertiesOfOptionPrices},
one can show that, under some additional assumptions, this solution can be treated as the classical one.
Then we can formulate the sufficient locally convexity preserving conditions for
the infinitesimal preservation of convexity at some point. This characterisation is given in
terms of a differential inequality on the coefficients of the considered operator.
It also allows to prove the convexity of $V^{\omega}_{\text{\rm E}}(s, t)$.
Then, in Theorem \ref{mainTheoremAmericanOptionConvexity} and Lemma \ref{TheoremAboutConvexityOfBermudan} we apply the dynamic programming principle (see \cite{propertiesOfAmericanOption})
in order to generalise the convexity property of $V^{\omega}_{\text{\rm E}}(s, t)$ to the value function $V^{\omega}_{\text{\rm A}}(s)$.

Later we focus on the American put option with the value function
\begin{equation*}\label{valueFunctionOmegaAmericanPutOption1}
V^{\omega}_{\text{\rm A}^{\text{\rm Put}}}(s) := \sup_{\tau\in\mathcal{T}}\mathbb{E}_{s}\left[e^{-\int_{0}^{\tau}\omega(S_w) dw} (K-S_{\tau})^{+}\right],
\end{equation*}
where the payoff function $g(s)=(K-s)^+$ for some strike price $K>0$.
Using the convexity property mentioned above we can conclude that the optimal stopping rule is defined as the first entrance of the process $S_t$ to the interval $[l,u]$, that is,
\begin{equation}\label{OptimalTau}
\tau_{l,u}:=\inf\{t\geq 0: S_t\in [l,u]\}.\end{equation}
In the next step, one has to identify
\begin{equation}\label{valueFunctionOmegaAmericanPutOption2}
v^{\omega}_{\text{\rm A}^{\text{\rm Put}}}(s, l, u):=\mathbb{E}_{s}\left[e^{-\int_0^{\tau_{l,u}} \omega(S_w) dw} (K-S_{\tau_{l,u}})^{+}\right]\end{equation}
and take maximum over levels $l$ and $u$ to identify the optimal stopping rule $\tau^*$ and to find the value function $V^{\omega}_{\text{\rm A}^{\text{\rm Put}}}(s)$.
This is done for the geometric spectrally negative L\'evy process $S_t=e^{X_t}$ where $X_t$ is a spectrally negative L\'evy process
starting at $X_0=\log S_0=\log s$.
We recall that spectrally negative L\'evy processes do not have positive jumps.
Hence, in particular, our analysis could be applied for the Black-Scholes market where $X_t$
is a Brownian motion with a drift.
To execute this plan we express $v^{\omega}_{\text{\rm A}^{\text{\rm Put}}}(s, l, u)$
in terms of the laws of the first passage times
and then we use the fluctuation theory developed in \cite{LiBo}.
In the whole analysis the use of the change of measure technique
developed in \cite{exponentialMartingale} is crucial as well.

Optimal levels $l^*$ and $u^*$ of the stopping region $[l^*,u^*]$ and
the price $V^{\omega}_{\text{\rm A}^{\text{\rm Put}}}(s)$ of the American put option could be
found by application of the appropriate HJB equation.
We prove this HJB equation and the smooth fit condition relying
on the classical approach of \cite{Lamberton} and \cite{peskir}.

Finally, to find the price of the American call option {\bf we
prove the put-call symmetry} in our set-up.
The proof is based on the exponential change of measure introduced in
\cite{exponentialMartingale}.

We analyse in detail the {\bf Black-Scholes model} and the case when {\bf a logarithm of the asset price
is a geometric linear drift minus compound Poisson process with exponentially distributed jumps} and various discount functions $\omega$.
The first example shows changes of the American options' prices in a gaussian and continuous market while the latter is to model the market including downward shocks in the assets' behaviour.
In this paper we present some specific examples for these two cases.

Our paper seems to be the {\bf first one analysing the optimal problem of the form \eqref{mainProblem} in this generality for jump-diffusion processes.}
For the classical diffusion processes Lamberton in \cite{Lambertonreward} proved that the value function in \eqref{mainProblem} is
continuous and can be characterised as the unique solution to a variational inequality in the sense of
distributions. Another crucial paper for our considerations is \cite{BeibelL}
which introduced {\bf discounting via a positive continuous additive
functional} of the process $S_t$ and used  the approach of It\^{o}
and McKean \cite[Sec. 4.6]{ItoMcKean}
to characterise the value function.
Note that $t\rightarrow \int_0^t \omega(S_w)dw$ is indeed additive functional.
A similar problem was also considered in \cite{Long}.

If $\omega(s)=(\log s -\log K)^+$ for a strike $K$
then $\int_0^t\omega(S_w)dw=\int_0^t (X_w-\log K)^+dw$ which equals the area under
the trajectory of $(X_t-\log K)^+$. Therefore, in this special case we can talk about the so-called {\bf area options}; see \cite{Davydov} for details.

Another interesting paper of Rodosthenous and Zhang \cite{Rodos} concerns the optimal stopping of an American call option in a random
time-horizon under a geometric spectrally negative L\'evy model. The random time-horizon is modeled by an {\bf omega default clock}
which is in their case the first time when the {\bf occupation time} of the asset price below a fixed level $y$ exceeds an independent exponential random variable
with mean $1/\varrho$. This corresponds to the special case
of our discounting with $\omega(s) =r +\varrho\mathds{1}_{\{s\leq y\}}$,
where $r$ is a risk-free interest rate.

Similar discounting was analysed in
\cite{Detemple2, Linetsky} where {\bf American step options} were considered.
In this case $\omega(s)=\varrho\mathds{1}_{\{s\in A(H)\}}$ where $A(H)=\{s\geq 0: \pm (s-H)\geq 0 \}$.
Moreover, the payoff function of a step option is the same as the payoff of a vanilla
option, except that it is deflated by a factor $e^{-\int_0^t \omega(S_w)dw}$ when the
knock-out rate $\varrho$ is positive or inflated by the same factor
when the knock-in rate $\varrho$ is negative. In both cases, the
factor depends exponentially on the cumulative excursion time
above or below a given barrier during the entire life of the option.
This exponential functional can be hence interpreted as a {\bf knock-out (knock-in) discount factor}.
Step options are mainly traded over-the-counter. They can
serve as a benchmark for the analysis and the design of certain
classes of occupation-time derivatives and related structured products.

This idea can be further generalised in a sense that the discount function $\omega$ can be treated
in the realm of real options
as more general knock-out (knock-in) factor.
One of the advantages of this type of the functional discounting is that an option
buyer can customize the option by selecting an appropriate functional rate according to his
risk aversion and the degree of confidence how the asset price will look like
during the whole option’s life. From the risk management perspective we can still
hedge this option by trading the underlying asset and we can still identify the value of
the contract.
Additionally, since different
market participants can select different discount rates, short-term manipulation by traders is substantially more difficult
and therefore considering such options
may help reducing market volatility.

The pricing technique developed in this paper {\bf can be applied to a wide range of securities and financial
contracts where the discounting in the above vein is affected by the underlying asset price process}.
Apart from the above mentioned examples of options, one can consider for example
{\bf Executive Stock Options} (EOSs) in which the executive may exercise
ESOs prematurely and leave the firm if an interesting opportunity arises
or for diversification or liquidity reasons.
Hence this policy can be determined by publicly available information such as stock prices.
As Carr and Linetsky \cite{CarrLin} noted this option corresponds either to
$\omega(s) =\lambda_f +\lambda_e\mathds{1}_{\{s>K\}}$
or
$\omega(s) =\lambda_f +\lambda_e\mathds{1}_{\{\log s>\log K\}}$, where
$\lambda_f$ is a
constant intensity of early exercise or forfeiture due to the exogenous voluntary
or involuntary employment termination
and $\lambda_e$ is the constant intensity of the early exercise due to the executive’s
exogenous desire for liquidity or diversification.
Another relevant application concerns {\bf R\&D projects}. Here,
the likelihood of achieving success before a competitor can depend on
the ability of the firm to invest resources in the discovery process.
If performance is poor, for instance due to mismanagement,
then the firm does not invest resources in the discovery process.
In the opposite scenario, more resources are devoted to research
activities. Hence the price of this type of projects depends
on the path-depended discounting as well; see \cite{TrigTsek} for survey.

The convexity of the value function and convexity preserving property, which is a key ingredient of our analysis,
have been studied quite extensively, see e.g. \cite{2, 3, chen, 4, 6, 8, 9, 11} for diffusion models,
and \cite{convexityPreserving, tysk} for one-dimensional jump-diffusion models.

We model dynamic of the asset price in a financial market by the jump-diffusion process.
The reason to take into account more general class of stochastic processes of asset prices
than in the seminal Black-Scholes market is the empirical observation that the log-prices of stocks have a  heavier left tail than the normal distribution, on which the seminal Black-Scholes model is founded.
The introduction of jumps in the financial market dates back to 
\cite{merton2}, who added a compound Poisson process to the standard Brownian motion to better describe dynamic of the logarithm of stocks.
Since then, there have been many papers and books working in this set-up, see e.g. \cite{cont, Schoutens} and references therein.
In particular, \cite[Table 1.1, p. 29]{cont} gives many other reasons to consider this type of market.
Apart from the classical Black-Scholes market one can consider
the normal inverse Gaussian model of
\cite{B10}, the hyperbolic
model of
\cite{B42}, the variance gamma model of
\cite{B80}, the CGMY model of
\cite{B24}, and the tempered stable process analysed
in
\cite{boyarchenkolevendorskii, B68}.
American options in the jump-diffusion markets have been studied in many papers as well; see e.g.
\cite{Aase2010, alili, AsmussenAvramPistorius, Erik1, boyarchenkolevendorskii, chan, ChesneyJeanblanc, Klimsiak, Mordecki}.

Identifying the solution of the optimal stopping problem by solving the corresponding HJB equation (as it is done in this paper as well) has been widely used; see \cite{krylov, peskir} for details.
In the context of American options with the constant discount function both methods of variational inequalities and viscosity solutions to the boundary value problems in the spirit of Bensoussan and Lions \cite{BenLions} are also well-known; see
e.g. \cite{Lamberton, Pham, pham98}.

To determine the unknown boundary of stopping region usually the smooth fit condition is applied; see e.g.
\cite{Kyprianou2007, Lamberton} for the geometric L\'evy process of asset prices.
As Lamberton and Mikou \cite{Lamberton} and Kyprianou and Surya \cite{Kyprianou2007} showed
the continuous fit is always satisfied but not necessary the smooth fit property.
Therefore we
prove that appropriate regularities of the process $S_t$ at critical points mentioned already above give smooth paste conditions
which generalises the classical results derived by \cite{Lamberton, Kyprianou2007}.
What we want to underline here is that using our approach (proving convexity and maximising over ends $l$ and $u$ of the stopping interval $[l,u]$)
allows to avoid identifying critical points via smooth paste conditions.

Apart from this, the interval form of the stopping region (hence producing double-sided continuation region)
is much more rare. It might come for example from the fact that when at time $t=0$ the discount rate is negative.
Then it is worth to wait
since discounting might increase the profit from such option.
This phenomenon has been already observed for fixed negative discounting (see
\cite{battauz1, battauz2, battauz3, tumilewicz, xia}) or in the case of American capped options with positive interest rate (see \cite{broadie, detemple}).

In this paper we also prove that in this general setting of asset-dependent discounting,
one can express the price of the call option in terms of the price of the put option.
It is called the put-call symmetry (or put-call parity).
Our finding supplements
\cite{EberleinPapapantoleon, mordecki} who extend
to the L\'evy market the findings by
\cite{CarrChesney}.
An analogous result for the negative discount rate case was obtained in  \cite{battauz1, battauz2, battauz3, tumilewicz}.
A comprehensive review of the put-call duality for American options is given in \cite{Detemple}.
We also refer to \cite[Section 7]{Detemple2014} and other references therein
for a general survey on the American options in the jump-diffusion model.

The paper is organised as follows.
In Section \ref{sec:main} we introduce basic notations and assumptions
that we use throughout the paper and we give the main results of this paper.
In Section \ref{sec:examples}, we present a few examples for which we obtain the analytical form of the value function, i.e. for the Black-Scholes market with negative $\omega$ function and for the market with prices being modeled by geometric L\'evy process with zero volatility and downward exponential jumps and $\omega$ being a linear function.
Section \ref{sec:proofs} contains proofs of all relevant theorems.
We put into Appendix proofs of auxiliary lemmas.
The last section includes our concluding remarks.

\section{Main results}\label{sec:main}

\subsection{Jump-diffusion process}
In this paper we assume a jump-diffusion financial market defined formally as follows.
On the basic probability space we define a couple $(B_t, v)$ adapted to the filtration $\mathcal{F}_t$, where $B_t$ is a standard Brownian motion and $v=v(dt, dz)$ is an independent of $B_t$ homogeneous Poisson random measure
on $\mathbb{R}_0^{+} \times \mathbb{R}$ for $\mathbb{R}_0^{+}=[0,+\infty)$.
Then the stock price process $S_t$ solves the following stochastic differential equation
\begin{equation}\label{diffusionWithJumps}
dS_t = \mu(S_{t-}, t) dt + \sigma(S_{t-}, t)dB_t + \int_\mathbb{R} \gamma(S_{t-}, t, z) \tilde{v}(dt, dz),
\end{equation}
where
\begin{itemize}
\item $\tilde{v}(dt, dz) = (v-q)(dt, dz)$ is a compensated jump martingale random measure of $v$,
\item $v$ is a homogenous Poisson random measure defined on $\mathbb{R}_0^{+}\times\mathbb{R}$ with intensity measure
\begin{equation*}
q(dt, dz) = dt\; m(dz).
\end{equation*}
\end{itemize}

If additionally, the jump-diffusion process has finite activity of jumps, i.e. when
\begin{equation*}\label{lambda}
\lambda:=\int_{\mathbb{R}}m(dz)<\infty,
\end{equation*}
then
$N_t=v([0,t]\times \mathbb{R})$ is a Poisson process and
$m$ 
can be represented as
\begin{equation*}
m(dz) = \lambda \mathbb{P}\left(e^{Y_i}-1\in dz\right),
\end{equation*}
where $\{Y_i\}_{i\in\mathbb{N}}$ are
i.i.d. random variables  independent of $N_t$ with distribution $\mu_Y$.
Note that $B_t$ and $N_t$ are independent of each other as well.
When additionally $\mu(s,t)=\mu s$, $\sigma(s,t)=\sigma s$ and $\gamma(s,t,z)=sz$, then
the asset price process $S_t$ is the geometric L\'evy process, that is,
\begin{equation}\label{spectrallyNegativeSt}
S_t= e^{X_t},
\end{equation}
where $X_t$ is a L\'evy process starting at $x=\log s$ with a triple $(\zeta, \sigma, \Pi)$ for
\begin{equation}\label{Pi}
\zeta:=\mu -\frac{\sigma^2}{2},\quad
\Pi(dx):=\lambda\mu_Y (dx).
\end{equation}
This observation follows straightforward from It\^o's rule.

\subsection{Assumptions}

Before we present the main results of this paper, we
state now the assumptions on the model parameters used later on.
We denote $\mathbb{R}^{+}:=(0,+\infty)$.
If we talk about convexity and concavity we mean it in a weak sense allowing these functions
to be constants within some regions.

\medskip

{\bf Assumptions (A)}
{\it \begin{enumerate}[label=(A\arabic*)]
    \item[]
	\item The drift parameter $\mu$: $\mathbb{R}^{+}\times \mathbb{R}_0^{+}\rightarrow \mathbb{R}$ and the diffusion parameter $\sigma$: $\mathbb{R}^{+}\times \mathbb{R}_0^{+}\rightarrow \mathbb{R}$ are continuous functions, while the jump size $\gamma$: $\mathbb{R}^{+}\times\mathbb{R}_0^{+}\times\mathbb{R}\rightarrow\mathbb{R}$ is measurable and for each fixed $z\in\mathbb{R}$, the function $(s, t) \rightarrow \gamma(s, t, z)$ is continuous.
	\label{A1}
    \item There exists a constant $C>0$ such that
    $$
    \mu^2(s, t) + \sigma^2(s, t) + \gamma^2(s, t, z) \leq C s^2
    $$ for all $(s, t, z) \in \mathbb{R}^{+}\times\mathbb{R}_0^{+}\times\mathbb{R}$.
    \label{A2}
    \item There exists a constant $C>0$ such that
    $$
    |\mu(s_2, t)-\mu(s_1, t)| + |\sigma(s_2, t)-\sigma(s_1, t)| + |\gamma(s_2, t, z) - \gamma(s_1, t, z)| \leq C|s_2-s_1|
    $$ for all $(s, t, z) \in \mathbb{R}^{+}\times\mathbb{R}_0^{+}\times\mathbb{R}$.
    \label{A3}
    \item There exists a constant $C>-1$ such that
    $$
    \gamma(s, t, z) > Cs
    $$ for all $(s, t, z) \in \mathbb{R}^{+}\times\mathbb{R}_0^{+}\times\mathbb{R}$.
    \label{A4}
    \item $g(s)\in C_{\text{pol}}(\mathbb{R}^{+})$, where $C_{\text{pol}}(\mathbb{R}^{+})$ denotes the set of functions of at most polynomial growth.
    \label{A5}
    \item $\omega(s)$ is bounded from below.
    \label{A6}
\end{enumerate}}

Assumptions \ref{A1}, \ref{A2}, \ref{A3} guarantee that there exists a unique solution to \eqref{diffusionWithJumps}.
Moreover, \ref{A2} and \ref{A4} imply that
\begin{equation*}
\mathbb{P}(S_t\leq 0 \text{ for some }t\in\mathbb{R}_0^{+})=0
\end{equation*}
which is a natural assumption since the process $S_t$ describes the stock price dynamic and its value has to be positive.
Additionally, assumptions \ref{A5} and \ref{A6}
ensure that $V^{\omega}_{\text{\rm A}}(s)$ is finite.

\begin{remark}\label{geomass}\rm
Note that assumptions \ref{A1}--\ref{A4} are all satisfied for the geometric L\'evy process.
\end{remark}
\subsection{Convexity of the value function}
Our first main result concerns the convexity of the value function $V^{\omega}_{\text{\rm A}}(s)$.
\begin{theorem}\label{mainTheoremAmericanOptionConvexity}
Let Assumptions (A) 
hold. Assume that
the payoff function $g$ is convex, $\omega$ is concave, the stock price process $S_t$ follows \eqref{diffusionWithJumps} and the following inequalities are satisfied
\begin{equation}\label{inequalityOnGamma2}
\frac{\partial^2\gamma(s,t,z)}{\partial s^2} \geq 0,
\end{equation}
\begin{equation}\label{inequalityOnMuOmegaG2}
\left(\frac{\partial^2\mu(s, t)}{\partial s^2} - 2\frac{d\omega(s)}{ds}\right)\frac{\partial V^{\omega}_{\text{\rm E}}(s, t)}{\partial s}\geq 0,
\end{equation}
where $V^{\omega}_{\text{\rm E}}(s, t)$ is defined in \eqref{valueFunctionOmegaEuropeanOption}.
Then the value function $V^{\omega}_{\text{\rm A}}(s)$ is convex as a function of $s$.
\end{theorem}

The proof of the above theorem is given
in Section \ref{sec:proofs}.

\begin{remark}\label{pierwszauwaga}\rm
We give now sufficient conditions in terms of model parameters for \eqref{inequalityOnMuOmegaG2} to be satisfied.
If $S_t$ is the geometric L\'evy process (hence $\mu(s,t)=\mu s$, $\sigma(s,t)=\sigma s$ and $\gamma(s,t,z)=sz$)
then \eqref{inequalityOnGamma2} is satisfied.
Let additionally $g(s)=(K-s)^+$. Then our optimal stopping problem is equivalent to
pricing the perpetual American put option with functional discounting.
If $\omega$ is non-decreasing function then the function
$s\rightarrow V^{\omega}_{\text{\rm E}}(s, t)$ is non-increasing.
Hence in this case condition \eqref{inequalityOnMuOmegaG2} is satisfied as well.
Concluding, if $\omega$ is concave and non-decreasing, then the value function of the perpetual American put option in geometric L\'evy market is convex as a function of the initial asset price $s$.
\end{remark}

\begin{remark}\rm
We have engineered above assumptions to handle mainly the put option (not a call option).
The call option can be then handled via put-call symmetry proved in
Theorem \ref{putCallSymmetry}. Hence there is no need to provide
sufficient conditions for both cases.
\end{remark}

\subsection{American put option and the optimal exercise time}
Assume now the particular case of \eqref{mainProblem} with
the payoff function
\begin{equation*}
g(s) = (K-s)^{+},
\end{equation*}
that is, the value function $V^{\omega}_{\text{\rm A}}(s)$ describes the price of a perpetual American put option.
The value function for this special choice of payoff function is denoted by
\begin{equation}\label{valueFunctionOmegaAmericanPutOption}
V^{\omega}_{\text{\rm A}^{\text{\rm Put}}}(s) := \sup_{\tau\in\mathcal{T}}\mathbb{E}_{s}\left[e^{-\int_{0}^{\tau}\omega(S_w) dw} (K-S_{\tau})\right].
\end{equation}
Note that above we used the fact that the option will not be realised when it equals zero, hence the plus in the payoff function could be skipped.

From \cite[Thm. 2.7, p. 40]{peskir} it follows that the optimal stopping rule is of the form
\begin{equation*}\label{opttaupeskir}
\tau^*=\inf\{t\geq 0: V^{\omega}_{\text{\rm A}^{\text{\rm Put}}}(S_t)=(K-S_t)\}.
\end{equation*}
From Theorem \ref{mainTheoremAmericanOptionConvexity}
we know that $V^{\omega}_{\text{\rm A}^{\text{\rm Put}}}(s)$ is convex. Moreover, from the definition of the value function it follows that $V^{\omega}_{\text{\rm A}^{\text{\rm Put}}}(s)\geq (K-s)$. Having both these facts in mind, together with linearity of the payoff function, it follows that $V^{\omega}_{\text{\rm A}^{\text{\rm Put}}}(s)$ and $g$ can cross each other in at most two points.
This observation leads straightforward to the conclusion about the form of the stopping region.
We recall that in \eqref{OptimalTau} and \eqref{valueFunctionOmegaAmericanPutOption2} we introduced the entrance time $\tau_{l, u} = \inf\{t\geq 0: S_t \in[l, u]\}
$ into the interval $[l,u]$ and the corresponding value function
$v^{\omega}_{\text{\rm A}^{\text{\rm Put}}}(s, l, u) = \mathbb{E}_{s}\left [e^{-\int_0^{\tau_{l, u}}\omega(S_w)dw} (K-S_{\tau_{l, u}})\right]$, respectively.
\begin{theorem}\label{lemmaAboutOptimalStoppingRule}
Let assumptions of Theorem \ref{mainTheoremAmericanOptionConvexity} hold. Then, the value function defined in \eqref{valueFunctionOmegaAmericanPutOption} is equal to
\begin{equation*}
V^{\omega}_{\text{\rm A}^{\text{\rm Put}}}(s) = v^{\omega}_{\text{\rm A}^{\text{\rm Put}}}(s, l^{*}, u^{*}),
\end{equation*}
where
\begin{equation*}\label{lStaruStar}
v^{\omega}_{\text{\rm A}^{\text{\rm Put}}}(s, l^{*}, u^{*}):= \sup_{0\leq l\leq u\leq K}v^{\omega}_{\text{\rm A}^{\text{\rm Put}}}(s, l, u).
\end{equation*}
The optimal stopping rule is $\tau_{l^*, u^*}$, where $l^*, u^*$ realise the supremum above.
\end{theorem}

\begin{remark}\rm
Another characterisation of critical points $l^*$ and $u^*$ via smooth fit property is given
in Theorem \ref{smoothTheorem}.
\end{remark}

Theorem \ref{lemmaAboutOptimalStoppingRule} indicates that the optimal stopping rule in our problem is the first time when the process $S_t$ enters the interval $[l^*, u^*]$ for some $l^*\leq u^*$.
In the case when $l^*=u^*$ the interval becomes a point which is possible as well.
In some cases the above observation allows to identify the value function in a much more transparent way.
Finally, note that if the discount function $\omega$ is nonnegative, then it is never
optimal to wait to exercise the option for small asset prices, that is, always $l^*=0$ and the stopping region is one-sided.

\subsection{Spectrally negative geometric L\'evy process}\label{LevyModel}
We can express the value function
$V^{\omega}_{\text{\rm A}^{\text{\rm Put}}}(s)$ explicitly for the spectrally negative
geometric L\'evy process  defined in
\eqref{spectrallyNegativeSt}, that is when
\[S_t=e^{X_t},\]
where $X_t$ is a spectrally negative L\'evy process with $X_0 = x=\log s$ and hence $S_0 = s$.
This means that $X_t$ does not have positive jumps which corresponds to the inclusion of the support of L\'evy measure
$m$ on the negative half-line. This is very common assumption which is justified by some financial crashes; see
e.g. \cite{alili, avram,chan}.
One can easily observe that the dual case of spectrally positive L\'evy process $X_t$ can be also handled in a similar way. We decided
to skip this analysis and focus only on a more natural, from a practical perspective, spectrally negative scenario.
We express the value function in terms of some special functions, called omega scale functions; see \cite{LiBo} for details.

To introduce these functions let us define first the Laplace exponent via
\begin{equation*}\label{laplaceexponent}
\psi(\theta) := \frac{1}{t}\log \mathbb{E}[e^{\theta X_t}\mid X_0=0],
\end{equation*}
which is finite al least for $\theta\geq 0$ due to downward jumps.
This function is strictly convex, differentiable, equals zero at zero and tends to infinity at infinity.
Hence there exists its right inverse $\Phi(q)$ for $q\geq 0$.

%

The key functions for the fluctuation theory are the scale functions; see \cite{kuznetsov}.
The first scale function $W^{(q)}(x)$ is the unique right continuous function disappearing on the negative half-line whose Laplace transform is
\begin{equation}\label{scaleFunction}
\int_0^{\infty} e^{-\theta x} W^{(q)}(x) dx = \frac{1}{\psi(\theta)-q}
\end{equation}
for $\theta>\Phi(q)$.

For any measurable function $\xi$ we define the $\xi$-scale functions
$\{\mathcal{W}^{(\xi)}(x), x\in\mathbb{R}\}$, $\{\mathcal{Z}^{(\xi)}(x), x\in\mathbb{R}\}$ and $\{\mathcal{H}^{(\xi)}(x), x\in\mathbb{R}\}$
as the unique solutions to the following renewal-type equations
\begin{align}
\mathcal{W}^{(\xi)}(x) &= W(x) + \int_0^{x} W(x - y)\xi(y)\mathcal{W}^{(\xi)}(y)dy,\label{scaleFunctionW} \\
\mathcal{Z}^{(\xi)}(x) &= 1 + \int_0^{x} W(x - y)\xi(y)\mathcal{Z}^{(\xi)}(y)dy\label{scaleFunctionZ},\\
\mathcal{H}^{(\xi)}(x) &= e^{\Phi(c)x} + \int_0^x W^{(c)}(x - z)(\xi(z) - c)\mathcal{H}^{(\xi)}(z) dz,\label{scaleH}
\end{align}
where $W(x)=W^{(0)}(x)$ is a classical zero scale function and in equation \eqref{scaleH} it is additionally assumed that
$\xi(x)=c$ for all $x\leq 0$ and some constant $c\in \mathbb{R}$.
We also define the function $\{\mathcal{W}^{(\xi)}(x,z), (x,z)\in\mathbb{R}^2\}$ solving the following equation
\begin{align}
\mathcal{W}^{(\xi)}(x,z) &= W(x-z) + \int_z^{x} W(x - y)\xi(y)\mathcal{W}^{(\xi)}(y,z)dy.\label{scaleFunctionWb}
\end{align}
We introduce the following $S_t$ counterparts of the scale functions \eqref{scaleFunctionW}, \eqref{scaleFunctionZ}, \eqref{scaleH} and \eqref{scaleFunctionWb}
\begin{align}
\mathscr{W}^{(\xi)}(s) &:= \mathcal{W}^{(\xi\circ {\rm exp})}(\log s),\label{scaleFunctionW2} \\
\mathscr{Z}^{(\xi)}(s) &:= \mathcal{Z}^{(\xi\circ {\rm exp})}(\log s), \label{scaleFunctionZ2} \\
\mathscr{H}^{(\xi)}(s) &:= \mathcal{H}^{(\xi\circ {\rm exp})}(\log s),\label{scaleH2} \\
\mathscr{W}^{(\xi)}(s,z) &:= \mathcal{W}^{(\xi\circ {\rm exp})}(\log s,z),\label{scaleFunctionW2b}
\end{align}
where $\xi\circ {\rm exp} (x):= \xi(e^x)$.

For $\alpha$ for which the Laplace exponent is well-defined we can define a new probability measure $\mathbb{P}^{(\alpha)}$ via
\begin{equation}\label{changemeas}
\left.\frac{d\mathbb{P}^{(\alpha)}_{s}}{d\mathbb{P}_{s}}\right\vert_{\mathcal{F}_t} = e^{\alpha (X_{t}-\log s)-\psi(\alpha)t}.
\end{equation}
By \cite{exponentialMartingale} and \cite[Cor. 3.10]{kyprianou}, under $\mathbb{P}^{(\alpha)}$, the process $X_t$ is again spectrally negative L\'evy process with the Laplace exponent of the form
\begin{equation}\label{changepsi}
\psi^{(\alpha)}(\theta):=\psi(\theta+\alpha)-\psi(\alpha).
\end{equation}
For this new probability measure $\mathbb{P}^{(\alpha)}$ we can define
$\xi$-scale functions
which are denoted by adding subscript $\alpha$ to the regular counterparts, hence we have
$\mathscr{W}^{(\xi)}_{\alpha}(s)$, $\mathscr{Z}^{(\xi)}_{\alpha}(s)$, $\mathscr{H}^{(\xi)}_{\alpha}(s)$ and $\mathscr{W}^{(\xi)}_{\alpha}(s, z)$.

We additionally define the following functions
\begin{align*}\label{eta}
&
\omega_u(s):=\omega(su) \quad \text{and}\quad
\omega_u^\alpha(s):=\omega_u(s)- \psi(\alpha).
\end{align*}
The main result is given in terms of the resolvent density at $z$ of $X_t$ starting at $\log s-\log u$
killed by the potential $\omega_u$ and on exiting from positive half-line given by
\begin{equation}\label{resovent}
r(s,u, z) := \mathscr{W}^{(\omega_u)}(\log s-\log u) c_{\mathscr{W}^{(\omega_u)}/\mathscr{W}^{(\omega_u)}}(z)- \mathscr{W}^{(\omega_u)}(\log s-\log u, z),
\end{equation}
where
\[c_{\mathscr{W}^{(\omega_u)}/\mathscr{W}^{(\omega_u)}}(z):=
\lim_{y\rightarrow\infty}\frac{\mathscr{W}^{(\omega_u)}(\log y, z)}{\mathscr{W}^{(\omega_u)}(\log y)}.\]

\begin{theorem}\label{mainTheoremGeometricLevy}
Assume that the stock price process $S_t$ is described by \eqref{spectrallyNegativeSt} with $X_t$ being the spectrally negative L\'evy process and $\omega$ is a measurable, bounded from below, concave and non-decreasing function such that
\begin{equation}\label{etaassump}\omega(s)=c \text{ for all $s\in(0, 1]$ and some constant $c\in \mathbb{R}$.}
\end{equation}
Then 
\begin{gather*}
\begin{split}
v^{\omega}_{\text{\rm A}^{\text{\rm Put}}}(s, l, u) &= \frac{\mathscr{H}^{(\omega)}(s)}{\mathscr{H}^{(\omega)}(l)}(K-l)\mathds{1}_{\{s<l\}} + (K-s)\mathds{1}_{\{s\in[l, u]\}} \\ &+
\Bigg\{ \int_0^{\infty} \int_0^{\infty} \frac{\mathscr{H}^{(\omega_u)}((\frac{u}{e^y})\wedge l)}{\mathscr{H}^{(\omega_u)}(l)} (K-e^{\log l \vee (\log u-y)}) r(s,u, z)\Pi(-z-dy) dz  \\ &+ (K-u)
\left(\lim_{\alpha\rightarrow\infty} \left(\frac{s}{u}\right)^{\alpha}\left(\mathscr{Z}^{(\omega_u^\alpha)}_{\alpha}\left(\frac{s}{u}\right) - c_{\mathscr{Z}^{(\omega_u^\alpha)}_{\alpha}/\mathscr{W}^{(\omega_u^\alpha)}_{\alpha}}\mathscr{W}^{(\omega_u^\alpha)}_{\alpha}\left(\frac{s}{u}\right)\right)\right) \Bigg\} \mathds{1}_{\{s>u\}},
\end{split}
\end{gather*}
where
\begin{equation*}\label{climitmath}
\begin{aligned}
c_{\mathscr{Z}^{(\omega_u^\alpha)}_{\alpha}/\mathscr{W}^{(\omega_u^\alpha)}_{\alpha}} &= \lim_{z\rightarrow\infty} \frac{\mathscr{Z}^{(\omega_u^\alpha)}_{\alpha}(z)}{\mathscr{W}^{(\omega_u^\alpha)}_{\alpha}(z)}
\end{aligned}
\end{equation*}
and $r(s, u, z)$ is given in \eqref{resovent}.\\
If $l=0$ then assumption \eqref{etaassump} is superfluous and
\begin{gather*}
\begin{split}
v^{\omega}_{\text{\rm A}^{\text{\rm Put}}}(s, 0, u) &= (K-s)\mathds{1}_{\{s\in[0, u]\}} \\ &+
\Bigg\{ \int_0^{\infty} \int_0^{\infty}  (K-e^{\log u-y}) r(s,u, z)\Pi(-z-dy) dz  \\ &+ (K-u)
\left(\lim_{\alpha\rightarrow\infty} \left(\frac{s}{u}\right)^{\alpha}\left(\mathscr{Z}^{(\omega_u^\alpha)}_{\alpha}\left(\frac{s}{u}\right) - c_{\mathscr{Z}^{(\omega_u^\alpha)}_{\alpha}/\mathscr{W}^{(\omega_u^\alpha)}_{\alpha}}\mathscr{W}^{(\omega_u^\alpha)}_{\alpha}\left(\frac{s}{u}\right)\right)\right) \Bigg\} \mathds{1}_{\{s>u\}}.
\end{split}
\end{gather*}

\end{theorem}

The proof of the above theorem is given
in Section \ref{sec:proofs}.
\begin{remark}\rm
For the general case when $l>0$ the assumption \eqref{etaassump} is a technical one and it is a consequence of the assumption made in
\cite[Thm. 2.5]{LiBo} which is used in the proof. We believe that this assumption is superfluous though.
\end{remark}

\subsection{HJB, smooth and continuous fit properties}
The classical approach via HJB system is possible in our set-up as well.
More precisely, as before in \eqref{spectrallyNegativeSt} we have
\[S_t=e^{X_t}\]
for the L\'evy process $X_t$ with the triple ($\zeta, \sigma, \Pi)$.
We start from the observation that using \cite[Thm. 31.5, Chap. 6]{Sato}
and It\^o's formula one can conclude that the process $S_t$ is a Markov process with an infinitesimal generator
\begin{equation*}
\mathcal{A}f(s) = A^C f(s) +  A^J f(s),
\end{equation*}
where $A^C$ is the linear second-order differential operator of the form
\begin{equation*}
A^C f(s) = \frac{\sigma^2 s^2}{2} f^{\prime\prime}(s)
+ \left(\zeta +\frac{\sigma^2}{2}\right)s f^\prime (s)
\end{equation*}
and $A^J$ is the integral operator given by
\begin{equation*}
A^J f(s) = \int_{(-\infty, 0)}
\left(f(s e^z) - f(s) - s|z|f^\prime(s)\mathds{1}_{\{|z|\leq 1\}} \right) \Pi(dz).
\end{equation*}
The domain $D(\mathcal{A})$ of this generator consists of the functions belonging to $C^{2}(\mathbb{R}^+)$ if $\sigma >0$ and $C^{1}(\mathbb{R}^+)$ if $\sigma=0$.
In this paper we prove that $V^{\omega}_{\text{\rm A}}(s)$ satisfies the HJB equation given below with
appropriate smooth fit conditions.
We recall that $1$ is regular for $(0, 1)$ and for the process $S_t$
if $\mathbb{P}_1(\tau_{(0, 1)}=0)=1$ for $\tau_{(0, 1)}=\inf\{t> 0: S_t\in (0, 1)\}$.
Similarly, we can define regularity for $(1, +\infty)$.
Note that regularity of $S_t$ at $1$ corresponds to regularity of $X_t$ at $0$
for the negative or positive half-line.
\begin{theorem}\label{smoothTheorem}
Let $\omega$ be a bounded from below and concave function with the opposite monotonicity to the payoff function $g$.
Assume that $V^{\omega}_{\text{\rm A}}(s)\in D(\mathcal{A})$ and $g(s)\in C^1(\mathbb{R}^+)$.
Then $V^{\omega}_{\text{\rm A}}(s)$ solves uniquely the following HJB system
\begin{equation}\label{CauchyProblem1}
  \begin{cases}
    \mathcal{A} V^{\omega}_{\text{\rm A}}(s) - \omega(s) V^{\omega}_{\text{\rm A}}(s)= 0, \quad & \quad s\notin [l^*,u^*], \\
	V^{\omega}_{\text{\rm A}}(s) = g(s), \quad &\quad  s\in [l^*,u^*].
  \end{cases}
\end{equation}
Moreover, if $1$ is regular for $(0, 1)$ and for the process $S_t$ then
there is a smooth fit at the right end of the stopping region
\begin{equation*}\label{smooth1}
(V^{\omega}_{\text{\rm A}})^\prime(u^*)=g^\prime(u^*).
\end{equation*}
Similarly, if $1$ is regular for $(1, +\infty)$ and for the process $S_t$ then
there is a smooth fit at the left end of the stopping region
\begin{equation*}\label{smooth2}
(V^{\omega}_{\text{\rm A}})^\prime(l^*)=g^\prime(l^*).
\end{equation*}
\end{theorem}
The proof of the above theorem is given
in Section \ref{sec:proofs}.

\begin{remark}\label{smoothvalue}
\rm Let us consider the American put option. Then from Theorems \ref{lemmaAboutOptimalStoppingRule} and  \ref{mainTheoremGeometricLevy}, we can conclude
that smoothness of the value function $V^{\omega}_{\text{\rm A}^{\text{\rm Put}}}(s)$ corresponds to the smoothness
of the $\xi$-scale functions for $\omega$, $\omega_u$ and $\omega_u^\alpha$.
From the definition of these functions given in \eqref{scaleFunctionW}, \eqref{scaleFunctionZ} and \eqref{scaleH} it follows that
the smoothness of the latter functions is equivalent to the smoothness of the first scale function observed under measures
$\mathbb{P}$ and
$\mathbb{P}^{(\alpha)}$.
By \cite[Lem. 8.4]{kyprianou} the smoothness of the first scale function does not change under the exponential change of measure \eqref{changemeas}.
Thus from \cite[Lem. 2.4, Thms 3.10 and 3.11]{kuznetsov} if follows that
\begin{itemize}
\item if $\sigma >0$ then $V^{\omega}_{\text{\rm A}^{\text{\rm Put}}}(s)\in C^{2}(\mathbb{R}^+)$;
\item if $\sigma =0$ and the jump measure $\Pi$ is absolutely continuous or $\int_{-1}^0|x|\Pi(dx)=+\infty$, then  $V^{\omega}_{\text{\rm A}^{\text{\rm Put}}}(s)\in C^{1}(\mathbb{R}^+)$.
\end{itemize}
Moreover, by \cite[Prop. 7]{alili}, $1$ is regular for both $(0, 1)$ and $(1,+\infty)$ if $\sigma >0$.
Hence HJB system \eqref{CauchyProblem1} with the smooth fit property could be used without any additional assumptions as long as $\sigma >0$.
If one has single continuation region $[u^*, +\infty)$ and $\sigma=0$ then by \cite[Prop. 7]{alili}
to get the smooth fit condition at $u^*$ it is sufficient to assume that the drift $\zeta$ of the process $X_t$ is strictly negative.
\end{remark}

\subsection{Put-call symmetry}
The put-call parity allows to calculate the American call option price having the put one.
We formulate this relation again for
$S_t$ being a general geometric L\'evy process defined in \eqref{spectrallyNegativeSt}, that is, $S_t=e^{X_t}$ for $X_t$
being a general spectrally negative L\'evy process
having triple
\[(\zeta, \sigma, \Pi)\]
for $\zeta$ and $\Pi$ defined in \eqref{Pi} and starting position $X_0=\log S_0=s$.
Apart from the function
\begin{equation*}\label{putSymmetry}
v^{\omega}_{\text{\rm A}^{\text{\rm Put}}}(s, K, \zeta, \sigma, \Pi, l, u) := \mathbb{E}_{s}[e^{-\int_0^{\tau_{l, u}}\omega(S_w)dw} (K-S_{\tau_{l, u}})^{+}]
\end{equation*}
defined in \eqref{valueFunctionOmegaAmericanPutOption2} we denote
\begin{equation*}
v^{\omega}_{\text{\rm A}^{\text{\rm Call}}}(s, K, \zeta, \sigma, \Pi, l, u) := \mathbb{E}_{s}[e^{-\int_0^{\tau_{l, u}}\omega(S_w)dw} (S_{\tau_{l, u}} - K)^{+}].
\end{equation*}

\begin{theorem}\label{putCallSymmetry}
Assume that
$\psi(1)=\log \mathbb{E}e^{X_1}=\log \mathbb{E} S_1$ is finite. Let $l \leq u \leq K$.
Then we have
\begin{equation}\label{pierwszatozsamosc}
v^{\omega}_{\text{\rm A}^{\text{\rm Call}}}(s, K, \zeta, \sigma, \Pi, l, u) = v^{\vartheta^{(1)}}_{\text{\rm A}^{\text{\rm Put}}}\left(K, s, -\zeta, \sigma, \hat{\Pi}, \frac{s}{u}K, \frac{s}{l}K\right),
\end{equation}
where
\begin{align}
\hat{\Pi}(dx) &:= e^{-x}\Pi(-dx), \label{hatPi}\\
\vartheta^{(1)}(\cdot) &:= \omega\left(\frac{1}{\cdot}sK\right) - \psi(1).\nonumber
\end{align}
Moreover, if assumptions of Theorem \ref{mainTheoremAmericanOptionConvexity} hold for the function
$\vartheta^{(1)}$ then
the American call option admits a double continuation region with optimal stopping boundaries
$l^*_c$ and $u^*_c$ such that
\begin{equation}\label{symm}
l^*_cu^*=l^*u_c^*=sK,
\end{equation}
where $l^*$ and $u^*$ are the stopping boundaries for the put option.
\end{theorem}
The proof of the above theorem is given
in Section \ref{sec:proofs}.

\begin{remark}\rm
Note that the value function of the American call option is expressed in terms of the American put option
calculated for the L\'evy process $\hat{X}_t$ being dual to $X_t$ process observed under the measure
$\mathbb{P}^{(1)}$. In particular, the jumps of the process $\hat{X}_t$ have the opposite direction to the jumps of the
process $X_t$ for which the put option is priced.
In general, determining the conditions for $\omega$ such that $\vartheta^{(1)}$
satisfies all the assumptions of Theorem \ref{mainTheoremAmericanOptionConvexity} seems to be impossible and then we can only work on a case-by-case basis.
\end{remark}


\subsection{Black-Scholes model}\label{BSModel}
We can give more detailed analysis in the case of Black–Scholes model in which
the stock price process $S_t=e^{X_t}$, where
\begin{equation}\label{X_tBS}
X_t = \log s + \zeta t + \sigma B_t
\end{equation}
with $\zeta = \mu-\frac{\sigma^2}{2}$, while $\mu\in \mathbb{R}$ and $\sigma>0$ are the parameters called drift and volatility, respectively.
Under the martingale measure we have $\mu = r$, where $r$ is a risk-free interest rate.

\begin{theorem}\label{mainTheoremBSModel}
Assume that $\omega$ is a bounded from below, concave and non-decreasing function. For Black-Scholes model \eqref{spectrallyNegativeSt} with $X_t$ given in \eqref{X_tBS} the function $v^{\omega}_{\text{\rm A}^{\text{\rm Put}}}(s, l, u)$ defined in \eqref{valueFunctionOmegaAmericanPutOption2} is given by
\begin{gather*}\label{valueFunctionBS}
\begin{split}
v^{\omega}_{\text{\rm A}^{\text{\rm Put}}}(s, l, u) &= \frac{h(s)}{h(l)}(K-l)\mathds{1}_{\{s<l\}} + (K-s)\mathds{1}_{\{s\in[l, u]\}} \\ &+ \frac{h(s)}{h(u)}(K-u)\mathds{1}_{\{s>u\}},
\end{split}
\end{gather*}
where $h(s)$ is a solution to
\begin{equation}\label{ODEForh}
\frac{\sigma^2 s^2}{2}h^{\prime\prime}(s) + \mu s h^\prime (s) - \omega(s)h(s) = 0
\end{equation}
which satisfies
\begin{equation}\label{ConditionsOnh(x)2}
\begin{cases}
h(s) = g(s), \quad s \in [l^{*}, u^{*}],\\
\displaystyle\lim_{s\rightarrow\infty}h(s) = \text{\rm const\;.}\\
\end{cases}
\end{equation}
\end{theorem}
The proof of the above theorem is given
in Section \ref{sec:proofs}.

\begin{remark}\label{remarkSmooth}\rm
The optimal boundaries $l^*$ and $u^*$ can be found from the smooth fit property
given in Theorem \ref{smoothTheorem}.
\end{remark}

\subsection{Exponential crashes market}\label{BSModelExpoJump}
We can construct more explicit equation for the value function for the case of Black–Scholes model with additional downward exponential jumps, that is, as in \eqref{spectrallyNegativeSt}, $S_t=e^{X_t}$
for
\begin{equation}\label{pertprocess}
X_t = \log s + \zeta t + \sigma B_t - \sum_{i=1}^{N_t} Y_i,
\end{equation}
where $\zeta = \mu-\frac{\sigma^2}{2}$, $N_t$ is the Poisson process with intensity $\lambda >0$ independent of Brownian motion $B_t$ and
$\{Y_i\}_{i\in \mathbb{N}}$ are i.i.d. random variables independent of $B_t$ and $N_t$ having exponential distribution with mean
$1/\varphi>0$. Moreover, under the martingale measure we obtain that $\mu = r + \frac{\lambda}{\varphi+1}$ with $r$ being a risk-free interest rate. The Laplace exponent of $X_t$ starting at $0$ is as follows
\begin{equation}\label{CLLaplace}
\psi(\theta)=\zeta\theta +\frac{\sigma^2}{2}\theta^2 - \frac{\lambda\theta}{\varphi+\theta}.
\end{equation}

For this model the price of American put option is easier to determine.
\begin{theorem}\label{mainTheoremBSModelExpojump0}
Assume that $\omega$ is a nonnegative, concave and non-decreasing function. For geometric L\'evy model \eqref{spectrallyNegativeSt} with $X_t$ given in \eqref{pertprocess} we have $l^*=0$. Furthermore,\\
(i) if $\sigma=0$ then
\begin{align}
V^{\omega}_{\text{\rm A}^{\text{\rm Put}}}(s):=\sup_{u>0}v^{\omega}_{\text{\rm A}^{\text{\rm Put}}}(s, 0, u)= \sup_{u>0}\bigg\{\left(K -
\frac{u\varphi}{\varphi+1}\right)\left(\mathscr{Z}^{(\omega_u)}\left(\frac{s}{u}\right)- c_{\mathscr{Z}^{(\omega_u)}/\mathscr{W}^{(\omega_u)}} \mathscr{W}^{(\omega_u)}\left(\frac{s}{u}\right)\right)\bigg\}
\label{valueexpojumps},
\end{align}
where $\mathscr{W}^{(\omega_u)}\left(\frac{s}{u}\right)$ and $\mathscr{Z}^{(\omega_u)}\left(\frac{s}{u}\right)$
are given in \eqref{scaleFunctionW2} and \eqref{scaleFunctionZ2}, respectively,
and
\begin{equation}\label{climitmath2}
\begin{aligned}
c_{\mathscr{Z}^{(\omega_u)}/\mathscr{W}^{(\omega_u)}} &:= \lim_{z\rightarrow\infty} \frac{\mathscr{Z}^{(\omega_u)}(z)}{\mathscr{W}^{(\omega_u)}(z)}.
\end{aligned}
\end{equation}
The optimal boundary $u^*$ is determined by the continuous fit condition
\begin{equation}\label{continuousexpo}
V^{\omega}_{\text{\rm A}^{\text{\rm Put}}}(u^*) = K-u^{*}.
\end{equation}
\\
(ii) If $\sigma>0$ then
\begin{equation}\label{valueexpojumps2}
\begin{aligned}
&V^{\omega}_{\text{\rm A}^{\text{\rm Put}}}(s) := \sup_{u>0}v^{\omega}_{\text{\rm A}^{\text{\rm Put}}}(s, 0, u) = \sup_{u>0}\bigg\{\left(K -
\frac{u \varphi}{\varphi+1}\right)\left(\mathscr{Z}^{(\omega_u)}\left(\frac{s}{u}\right)- c_{\mathscr{Z}^{(\omega_u)}/\mathscr{W}^{(\omega_u)}} \mathscr{W}^{(\omega_u)}\left(\frac{s}{u}\right)\right)
\\ &+ (K-u)
\left(\lim_{\alpha\rightarrow\infty} \left(\frac{s}{u}\right)^{\alpha}\left(\mathscr{Z}^{(\omega_u^\alpha)}_{\alpha}\left(\frac{s}{u}\right) -
c_{\mathscr{Z}^{(\omega_u^\alpha)}_{\alpha}/\mathscr{W}^{(\omega_u^\alpha)}_{\alpha}}
\mathscr{W}^{(\omega_u^\alpha)}_{\alpha}\left(\frac{s}{u}\right)\right)\right)\bigg\},
\end{aligned}
\end{equation}
where $\mathscr{W}^{(\omega_u^{\alpha})}_{\alpha}\left(\frac{s}{u}\right)$ and $\mathscr{Z}^{(\omega_u^{\alpha})}_{\alpha}\left(\frac{s}{u}\right)$
are the scale functions \eqref{scaleFunctionW2} and \eqref{scaleFunctionZ2} taken under measure $\mathbb{P}^{(\alpha)}$
and
\begin{equation*}\label{climitmath2Alpha}
\begin{aligned}
c_{{\mathscr{Z}^{(\omega_u^{\alpha})}_{\alpha}}/{\mathscr{W}^{(\omega_u^{\alpha})}_{\alpha}}} &:= \lim_{z\rightarrow\infty} \frac{\mathscr{Z}^{(\omega_u^{\alpha})}_{\alpha}(z)}{\mathscr{W}^{(\omega_u^{\alpha})}_{\alpha}(z)}.
\end{aligned}
\end{equation*}
The optimal boundary $u^*$ is determined by the smooth fit condition
\begin{equation*}\label{smoothexpo}
(V^{\omega}_{\text{\rm A}^{\text{\rm Put}}})'(u^*) = -1.
\end{equation*}
\end{theorem}
The proof of the above theorem is given
in Section \ref{sec:proofs}.

From \eqref{changepsi} (see also \cite[Prop. 5.6]{exponentialMartingale}) we simply note that the
Laplace exponent of $X_t$ taken under $\mathbb{P}^{(\alpha)}$ is of the same form as \eqref{CLLaplace}, i.e.
\begin{equation}\label{CLLaplaceAlpha}
\psi^{(\alpha)}(\theta)=\zeta^{(\alpha)}\theta +\frac{{\sigma^{(\alpha)}}^2}{2}\theta^2 - \frac{\lambda^{(\alpha)}\theta}{\varphi^{(\alpha)}+\theta},
\end{equation}
where $\zeta^{(\alpha)} = \zeta + \sigma^2\alpha$,
${\sigma^{(\alpha)}} = \sigma$, $\lambda^{(\alpha)} = \frac{\lambda\varphi}{\varphi+\alpha}$ and $\varphi^{(\alpha)} = \varphi + \alpha$. Hence, finding the scale functions under $\mathbb{P}$ and $\mathbb{P}^{(\alpha)}$ works in the same manner. To do so,
we recall that in \eqref{scaleFunctionW2} and \eqref{scaleFunctionZ2} we introduced them via
regular $\xi$-scale functions, that is
$\mathscr{W}^{(\xi)}(s) = \mathcal{W}^{(\xi\circ {\rm exp})}(x)$ and
$\mathscr{Z}^{(\xi)}(s) = \mathcal{Z}^{(\xi\circ {\rm exp})}(x)$ for $x=\log s$. Therefore, to identify a closed form of \eqref{valueexpojumps} and \eqref{valueexpojumps2} it suffices to find $\xi$-scale functions $\mathcal{W}^{(\xi)}(x)$ and $\mathcal{Z}^{(\xi)}(x)$
for a given generic function $\xi$.
We recall that both $\xi$-scale functions are given as the solutions of renewal equations
\eqref{scaleFunctionW} and \eqref{scaleFunctionZ} formulated in terms of the classical scale function $W(x)$.
From the definition of the first scale function given in \eqref{scaleFunction}
with $q=0$ and from \eqref{CLLaplace} with $\sigma>0$ we derive
\begin{equation*}\label{exposcaale}
W(x)=\sum_{i=1}^3 \Upsilon_ie^{\gamma_i x},
\end{equation*}
where $\gamma_i$ solves
\begin{equation}\label{defgamma}
\psi(\gamma_i)=0
\end{equation}
and
\begin{equation*}\label{Bi}
\Upsilon_i:=\frac{1}{\varphi^\prime(\gamma_i)}.
\end{equation*}

Note that one of the solution to \eqref{defgamma} equals $0$, so we can set $\gamma_1 = 0$.
In turn, if $\sigma=0$ in \eqref{CLLaplace} then
\begin{equation*}\label{exposcaale2}
W(x)=\sum_{i=1}^2 \Upsilon_i e^{\gamma_i x}
\end{equation*}
with $\gamma_1 = 0$, $\gamma_2 = \frac{\lambda - \varphi\mu}{\mu}$, $\Upsilon_1 = -\frac{\varphi}{\lambda-\varphi\mu}$ and $\Upsilon_2 = \frac{\lambda}{\mu(\lambda-\varphi \mu)}$.
Next theorem provides the ordinary differential equations
whose solutions are $\mathcal{W}^{(\xi)}(x)$ and $\mathcal{Z}^{(\xi)}(x)$. We use this result in the next section where we provide specific examples.

\begin{theorem}\label{mainTheoremBSModelExpojump}
We assume that the function $\xi$ is continuously differentiable.
For geometric L\'evy model \eqref{spectrallyNegativeSt} with $X_t$ given in \eqref{pertprocess} we have\\
(i) If $\sigma=0$ then the function $\mathcal{W}^{(\xi)}(x)$ solves \\
\begin{equation}\label{scaleequation}
{\mathcal{W}^{(\xi)}}''(x) = \left((\Upsilon_1+\Upsilon_2)\xi(x)+\gamma_2\right){\mathcal{W}^{(\xi)}}'(x) +
\left((\Upsilon_1+\Upsilon_2)\xi'(x)-\gamma_2\Upsilon_1\xi(x)\right){\mathcal{W}^{(\xi)}}(x)
\end{equation}
with
\begin{equation}\label{fbW}
\begin{cases}
{\mathcal{W}^{(\xi)}}(0) = \Upsilon_1 + \Upsilon_2, \\
{\mathcal{W}^{(\xi)}}'(0) = (\Upsilon_1+\Upsilon_2)^2\xi(0) + \Upsilon_2\gamma_2.
\end{cases}
\end{equation}

Moreover, the function $\mathcal{Z}^{(\xi)}(x)$ solves the same equation \eqref{scaleequation} with
\begin{equation}\label{fbZ}
\begin{cases}
{\mathcal{Z}^{(\xi)}}(0) = 1, \\
{\mathcal{Z}^{(\xi)}}'(0) = (\Upsilon_1+\Upsilon_2)\xi(0).
\end{cases}
\end{equation}
\\
(ii) If $\sigma>0$ then
the function $\mathcal{W}^{(\xi)}(x)$ solves\\
\begin{equation}\label{scaleequation2}
\begin{split}
{\mathcal{W}^{(\xi)}}'''(x) &= \left(\gamma_2+\gamma_3\right){\mathcal{W}^{(\xi)}}''(x) \\&+ \left(\Upsilon_2(\gamma_2-\gamma_3)\xi(x)-\gamma_2\gamma_3-\gamma_3\Upsilon_1\xi(x)\right){\mathcal{W}^{(\xi)}}'(x) \\&+ \left(\Upsilon_2(\gamma_2-\gamma_3)\xi'(x)+\gamma_2\gamma_3 \Upsilon_1 \xi(x)-\gamma_3 \Upsilon_1\xi'(x)\right){\mathcal{W}^{(\xi)}}(x)
\end{split}
\end{equation}
with
\begin{equation*}\label{fbW2}
\begin{cases}
{\mathcal{W}^{(\xi)}}(0) = 0, \\
{\mathcal{W}^{(\xi)}}'(0) = \Upsilon_2\gamma_2 + \Upsilon_3\gamma_3, \\
{\mathcal{W}^{(\xi)}}''(0) = \Upsilon_2 {\gamma_2}^2 + \Upsilon_3{\gamma_3}^2.
\end{cases}
\end{equation*}

Moreover, the function $\mathcal{Z}^{(\xi)}(x)$ solves the same equation \eqref{scaleequation2} with
\begin{equation*}
\begin{cases}
{\mathcal{Z}^{(\xi)}}(0) = 1, \\
{\mathcal{Z}^{(\xi)}}'(0) = 0, \\
{\mathcal{Z}^{(\xi)}}''(0) = \Upsilon_2(\gamma_2-\gamma_3) \xi(0) - \gamma_3 \Upsilon_1 \xi(0).
\end{cases}
\end{equation*}

\end{theorem}
The proof of the above theorem is given
in Section \ref{sec:proofs}.

\begin{remark}\rm
Note that in our case $l^*=0$ and from Theorem \ref{mainTheoremBSModelExpojump} it follows that
assumption \eqref{etaassump} is not required.
\end{remark}


\section{Examples}\label{sec:examples}

In this section we present the analytical form of value function \eqref{valueFunctionOmegaAmericanPutOption} for the particular $\omega$ and for the Black-Scholes model and Black-Scholes model with zero volatility and downward exponential jumps. In the first scenario, we take into account only the case of negative $\omega$, while in the second example we focus on the positive $\omega$.

\subsection{Black-Scholes model revisited}
Let
\[\omega(s)= -\frac{C}{s+1} - D,\]
where $C$ and $D$ are some positive constants. Applying Theorem \ref{mainTheoremBSModel}, we obtain
\begin{equation*}\label{valueFunctionBSPuts^a}
v^{\omega}_{\text{\rm A}^{\text{\rm Put}}}(s, l, u) = \frac{h(s)}{h(l)}(K-l)\mathds{1}_{\{s\in(0, l)\}} + (K-s)\mathds{1}_{\{s\in[l, u]\}} + \frac{h(s)}{h(u)}(K-u)\mathds{1}_{\{s\in(u, +\infty)\}},
\end{equation*}
where $h$ is a solution to
\begin{equation}\label{ODEs^a}
\frac{\sigma^2 s^2}{2}h^{\prime\prime}(s) + \mu s h^\prime (s) - \left(-\frac{C}{s+1} - D\right) h(s) = 0
\end{equation}
which satisfies
\begin{equation}\label{boundaryBS}
\begin{cases}
h(s) = g(s), \quad s \in [l^{*}, u^{*}],\\
\displaystyle\lim_{s\rightarrow\infty}h(s) = \text{const}\;.\\
\end{cases}
\end{equation}
Firstly, we solve the above equation and then
we look for boundaries $l^{*}$ and $u^{*}$ such that  $v^{\omega}_{\text{\rm A}^{\text{\rm Put}}}(s, l^{*}, u^{*}) = \sup_{0\leq l\leq u\leq K} v^{\omega}_{\text{\rm A}^{\text{\rm Put}}}(s, l, u)$. We find them using the smooth fit conditions.
The general solution to \eqref{ODEs^a} is given by
\begin{equation}\label{solutionBSPut}
\begin{aligned}
h(s) = K_1 s^{d_1}{}_2F_1(a_1; b_1; c_1; -s) + K_2 s^{d_2} {}_2F_1(a_2; b_2; c_2; -s),
\end{aligned}
\end{equation}
where
$L := \frac{1}{2} - \frac{\mu}{\sigma^2}$, $M := \sqrt{L^2 - \frac{2D}{\sigma^2}}$, $G := \sqrt{L^2 - \frac{2(C+D)}{\sigma^2}}$, while $a_i := (-1)^{i+1}(M-G)$, $b_i := (-1)^i(M+G)$, $c_i := 1+2(-1)^i G$, $d_i := (-1)^iG+L$ for $i = 1, 2$ and $K_1$, $K_2$ are some constants.

Using formula \eqref{solutionBSPut} and the boundary conditions given in \eqref{boundaryBS}
we can identify the form of value function \eqref{valueFunctionOmegaAmericanPutOption}.
Since we consider the negative $\omega$ we obtain a double continuation region. We take one of the summand from \eqref{solutionBSPut} for $s\in(0, l^{*})$ and the second one for $s\in(u^{*}, +\infty)$. This choice is made in a such a way
that on the given interval we impose to have a greater function of these two.
Hence we derive
\begin{equation*}
V^{\omega}_{\text{\rm A}^{\text{\rm Put}}}(s) =
\begin{cases}
K_2 s^{d_2} {}_2F_1(a_2; b_2; c_2; -s), &s\in(0, l^{*}), \\
K-s, &s\in[l^{*}, u^{*}], \\
K_1 s^{d_1}
{}_2F_1(a_1; b_1; c_1; -s), &s\in(u^{*}, +\infty).
\end{cases}
\end{equation*}
Using the smooth and continuous fit properties we can find $K_1$ and $K_2$ and
show that $l^{*}$ and $u^{*}$ solve the following equation
\begin{align}\label{lAnduRoots}
1 + {}_2F_1(a_i; b_i; c_i; -s) K_i D_i + s^{d_i} P_i = 0,
\end{align}
where
\begin{align*}
K_i &:= (K-s) \frac{s^{-d_i}}{{}_2F_1(a_i; b_i; c_i; -s)}, \\
D_i &:=  d_i s^{d_i - 1}, \\
P_i &:= -\frac{a_i b_i {}_2F_1(a_i+1; b_i+1; c_i+1; -s)}{c_i}
\end{align*}
for $i=1, 2$.
We numerically calculate the roots of \eqref{lAnduRoots} for $i=1, 2$ and we assign the smaller result to $l^{*}$ and the greater one to $u^{*}$.

Let us assume the given set of parameters $C = 0.001$, $D=0.01$, $K=20$, $\mu=5\%$ and $\sigma=20\%$.
The above numerical procedure produces in this case $l^{*}\approx 7.23$ and $u^{*}\approx 8.34$. Figure \ref{Example-BS} presents the value function that then arises.

\begin{figure}[ht]
\centering
\includegraphics[width=13cm]{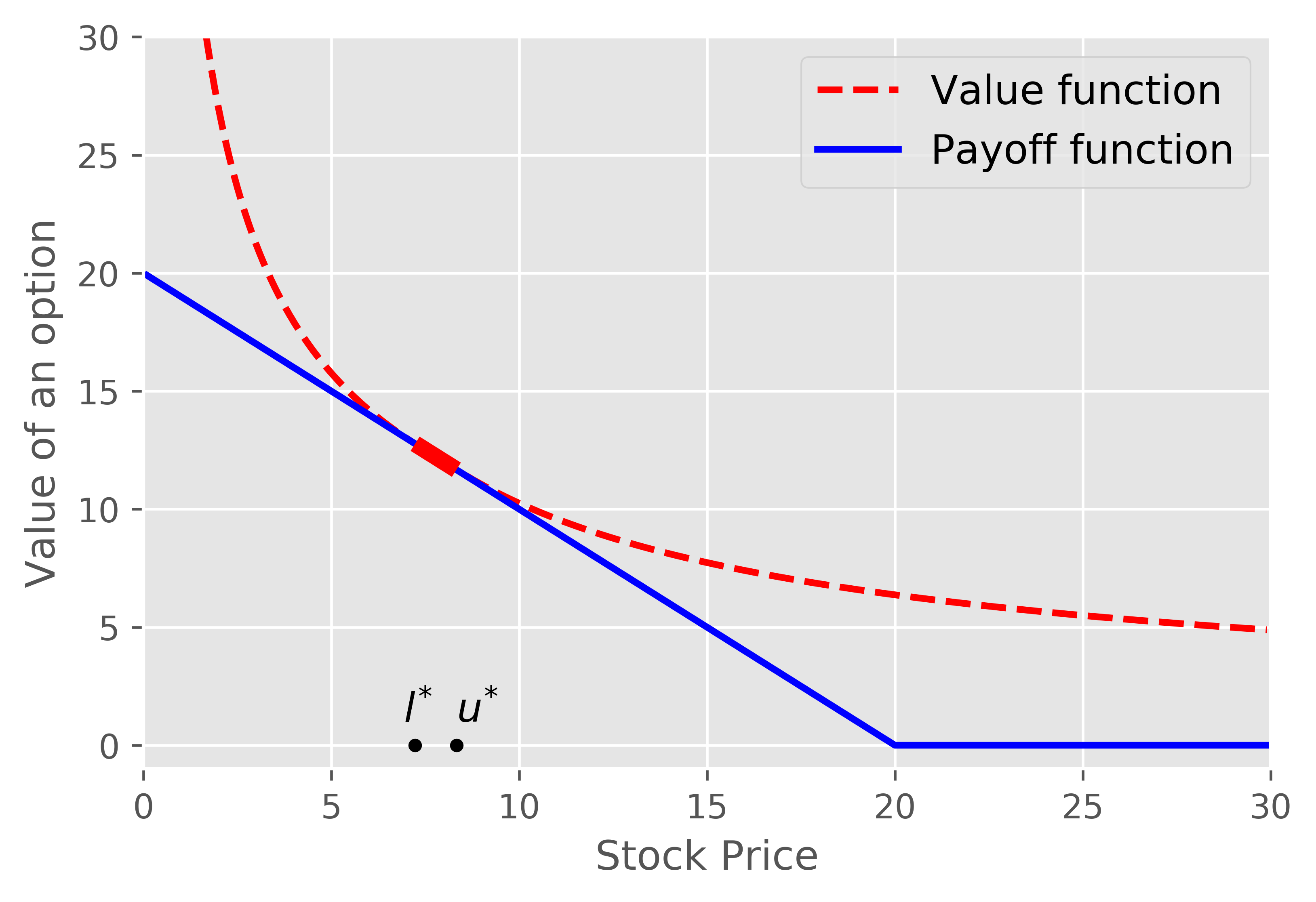}
\caption{The value and payoff functions for the given set of parameters: $C = 0.001$, $D=0.01$, $K=20$, $\mu=5\%$ and $\sigma=20\%$.}
\label{Example-BS}
\end{figure}

\begin{remark}\rm
Let us note that $\lim_{s\rightarrow 0^{+}} V^{\omega}_{\text{\rm A}^{\text{\rm Put}}}(s) = \infty$ which means that the price of the option is unlimited even for an arbitrarily low stock price. This is obviously a consequence of the fact that for $s\rightarrow 0^{+}$ the discount function is strictly negative.
\end{remark}

\subsection{Exponential crashes market revisited}
Let
\[\omega(s)= Cs,\]
where $C$ is some positive constant.
Note that this discount function is nonegative. Hence
from Theorem \ref{mainTheoremBSModelExpojump0}
it follows that $l^*=0$.
Let
\begin{align}\label{eta2}
&\eta(x) := \omega(e^x) = \omega(s) \quad\text{and}\quad \eta_u(x):=\eta(x+\log u).
\end{align}
From Theorem \ref{mainTheoremBSModelExpojump0} (see equation \eqref{valueexpojumps})
with $\sigma = 0$ and using \eqref{scaleFunctionW2} and \eqref{scaleFunctionZ2}
we can conclude that
\begin{equation*}
V^{\omega}_{\text{\rm A}^{\text{\rm Put}}}(s)=\sup_{u>0} v^{\omega}_{\text{\rm A}^{\text{\rm Put}}}(s, 0, u),
\end{equation*}
where
\begin{align*}\label{ValueFunctionBSWithJumps2}
v^{\omega}_{\text{\rm A}^{\text{\rm Put}}}(s, 0, u)= \left(K -
\frac{u\varphi}{\varphi+1}\right)\left(\mathcal{Z}^{(\eta_u)}\left(x-\log u\right)- c_{\mathcal{Z}^{(\eta_u)}/\mathcal{W}^{(\eta_u)}} \mathcal{W}^{(\eta_u)}\left(x-\log u\right)\right)
\end{align*}
and from \eqref{climitmath2}
\begin{equation*}
\begin{aligned}
c_{\mathscr{Z}^{(\omega_u)}/\mathscr{W}^{(\omega_u)}} &= c_{\mathcal{Z}^{(\eta_u)}/\mathcal{W}^{(\eta_u)}}:=\lim_{z\rightarrow\infty} \frac{\mathcal{Z}^{(\eta_u)}(z)}{\mathcal{W}^{(\eta_u)}(z)}.
\end{aligned}
\end{equation*}

From Theorem \ref{mainTheoremBSModelExpojump} it follows that ${\mathcal{W}^{(\eta)}}(x)$ solves
the following ordinary differential equation
\begin{equation}\label{ODE_BSwithJumps}
{\mathcal{W}^{(\eta)}}''(x) = (Ae^x+B){\mathcal{W}^{(\eta)}}'(x) + D e^x \mathcal{W}^{(\eta)}(x),
\end{equation}
with $A := \frac{C}{\mu}$, $B := \frac{\lambda-\varphi \mu}{\mu}$ and $D := C\frac{1 + \varphi}{\mu} $.
The above equation is also satisfied by ${\mathcal{Z}^{(\eta)}}(x)$.

From \eqref{fbW} and \eqref{fbZ} we conclude that
\begin{equation}\label{initialW}
\begin{cases}
{\mathcal{W}^{(\eta)}}(0) = \frac{1}{\mu}, \\
{\mathcal{W}^{(\eta)}}'(0) = \frac{C+\lambda}{\mu^2}.
\end{cases}
\end{equation}
and
\begin{equation}\label{initialZ}
\begin{cases}
{\mathcal{Z}^{(\eta)}}(0) = 1, \\
{\mathcal{Z}^{(\eta)}}'(0) = \frac{C}{\mu}.
\end{cases}
\end{equation}

The solution to \eqref{ODE_BSwithJumps} is of the form
\begin{equation*}\label{V(s)ForSigma0}
{\mathcal{W}^{(\eta)}}(x) = {\rm Re}\left\{K_1^W {_1F_1}\left(a_1; b_1; A e^x\right) + K_2^W (-A)^B e^{Bx} {_1F_1}\left(a_2; b_2; A e^x\right)\right\},
\end{equation*}
where $a_1 = \frac{D}{A}$,  $b_1 = 1 - B$, $a_2 = B + \frac{D}{A}$ and $b_2 = B + 1$,
while $K_1^W$ and $K_2^W$ are the constants that can be found based on the initial conditions \eqref{initialW}. A similar expression can be found for ${\mathcal{Z}^{(\eta)}}(x)$ with constants $K_1^Z$ and $K_2^Z$ satisfying \eqref{initialZ}.
Having the both forms of ${\mathcal{W}^{(\eta)}}(x)$ and ${\mathcal{Z}^{(\eta)}}(x)$ and by shifting these scale functions by $\log u$
we can produce figures of
$\mathcal{W}^{(\eta_u)}(x - \log u)$ and
$\mathcal{Z}^{(\eta_u)}(x - \log u)$.
Taking into account the asymptotic behaviour of ${_1F_1}\left(\cdot; \cdot; \cdot\right)$ we can calculate the constant $c_{\mathcal{Z}^{(\eta_u)}/\mathcal{W}^{(\eta_u)}}$. It has the following form
\begin{equation*}
c_{\mathcal{Z}^{(\eta_u)}/\mathcal{W}^{(\eta_u)}} = \frac{{\rm Re}\left\{K_1^Z \frac{\Gamma(b_1)}{\Gamma(a_1)}A^{a_1-b_1} + K_2^Z (-A)^B \frac{\Gamma(b_2)}{\Gamma(a_2)} A^{a_2-b_2}\right\}}{{\rm Re}\left\{K_1^W \frac{\Gamma(b_1)}{\Gamma(a_1)}A^{a_1-b_1} + K_2^W (-A)^B \frac{\Gamma(b_2)}{\Gamma(a_2)} A^{a_2-b_2}\right\}}.
\end{equation*}
Finally, using the continuous fit condition \eqref{continuousexpo}
we can find the optimal $u^{*}$.

Let us assume that $C = 0.1$, $K = 20$, $r = 5\%$, $\lambda = 6$, $\varphi = 2$.
The continuous fit property produces $u^{*}\approx 4.56$.
Figure \ref{Example-sigma0} shows the obtained value function.

\begin{figure}[ht]
\centering
\includegraphics[width=13cm]{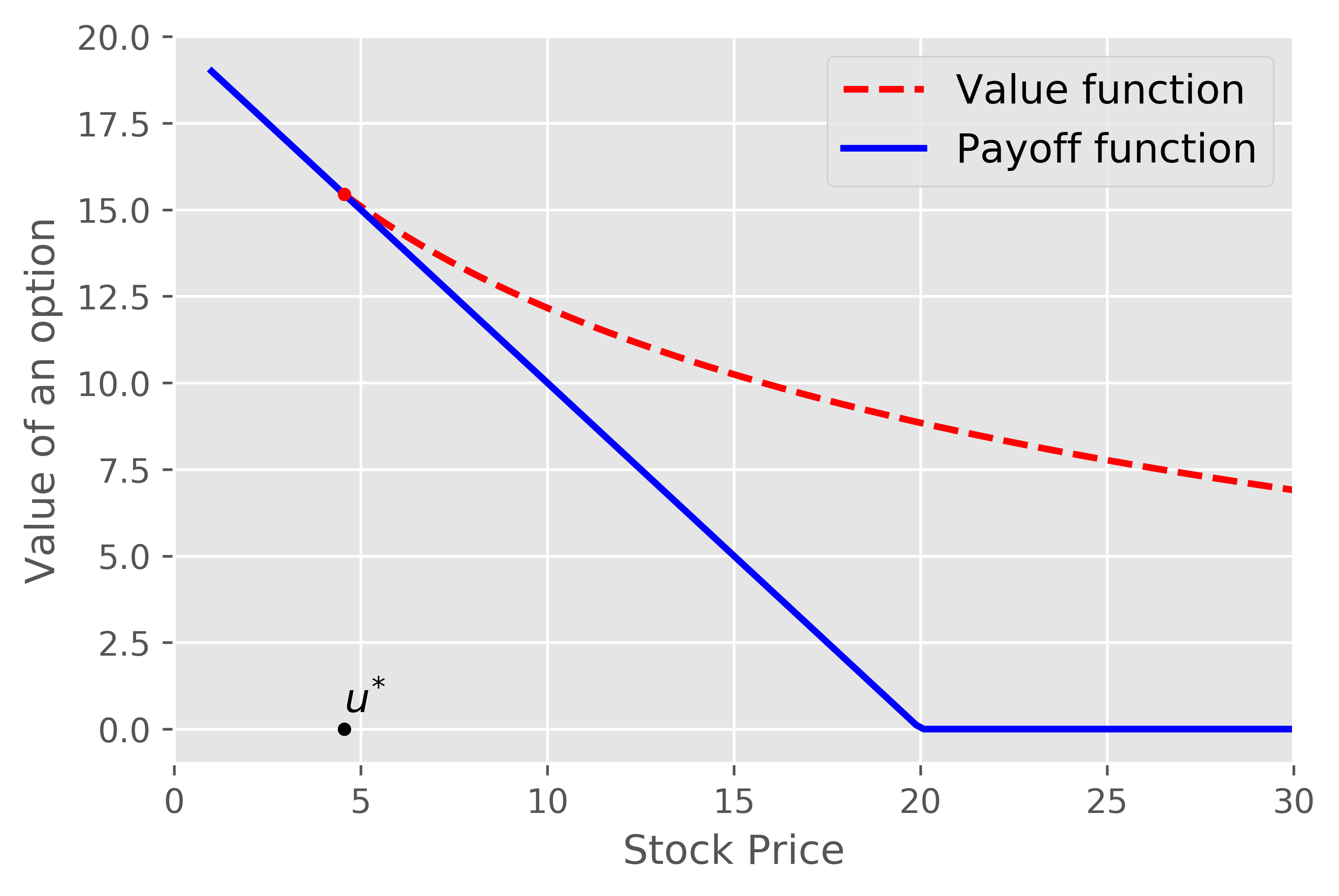}
\caption{The value and payoff functions for the given set of parameters: $C = 0.1$, $K = 20$, $r = 5\%$, $\lambda = 6$, $\varphi = 2$.}
\label{Example-sigma0}
\end{figure}

\begin{remark}\rm
In the example above we have provided the analytical form of $V^{\omega}_{\text{\rm A}^{\text{\rm Put}}}(s)$.
However, if we choose a more complex $\omega$ function and we are not be able to obtain the analytical form of $v^{\omega}_{\text{\rm A}^{\text{\rm Put}}}(s, 0,u)$ by solving a differential equation, we can try to solve it numerically to receive ${\mathcal{W}^{(\eta)}}(x)$, ${\mathcal{Z}^{(\eta)}}(x)$ and $c_{\mathcal{Z}^{(\eta_u)}/\mathcal{W}^{(\eta_u)}}$. In this example, we used both of these methods and we obtained the same result.
\end{remark}

\section{Proofs}\label{sec:proofs}
Before we prove main Theorem \ref{mainTheoremAmericanOptionConvexity} we show the convexity of European option price
$V^{\omega}_{\text{\rm E}}(s,t)$ defined in \eqref{valueFunctionOmegaEuropeanOption} as a function of $s$.
It is done in Theorems \ref{TheoremAboutConvexityOfValueFunctionDiffusionWithJumps} and \ref{TheoremAboutConvexityOfValueFunctionDiffusionWithJumps2}.
In the proof we apply idea demonstrated in \cite[Prop. 4.1]{propertiesOfOptionPrices}.
Later, in the proof of Theorem \ref{mainTheoremAmericanOptionConvexity}, we use a variant of the maximum principle.

Let us introduce a set $E\subset\mathbb{R}\times[0,T]$.
We use the following notations
\begin{itemize}
\item $C_{\alpha}(E)$ is the set of locally
H\"older($\alpha$) functions with $\alpha\in(0, 1)$,

\item $C_{\text{pol}}(E)$ is the set of functions of at most polynomial growth in $s$,

\item $C^{p, q}(E)$ is the set of functions for which all the derivatives $\frac{\partial^k}{\partial s^k}\left(\frac{\partial^l f(s,t)}{\partial t^l}\right)$ with $|k|+2l\leq p$ and $0\leq l\leq q$ exist in the interior of $E$ and have continuous extensions to $E$,

\item $C^{p, q}_{\alpha}(E)$ and $C^{p, q}_{\text{pol}}(E)$ are the sets of functions
$f\in C^{p, q}(E)$ for which all the derivatives $\frac{\partial^k}{\partial s^k}\left(\frac{\partial^l f(s,t)}{\partial t^l}\right)$ with $|k|+2l\leq p$ and $0\leq l\leq q$ belong to $C_{\text{pol}}(E)$ and $C_{\alpha}(E)$, respectively.
\end{itemize}
We need the following conditions in the proofs.
\medskip

{\bf Assumptions (B)}\\
{\it There exist constants $C > 0$ and $\alpha\in(0, 1)$ such that
\begin{enumerate}[label=(B\arabic*)]
    \item $\mu(s,t)\in C^{2,1}_{\alpha}(\mathbb{R}^{+}\times[0, T])$;
    \label{B1}
    \item $\sigma^2(s,t)\geq Cs^2$ for all $(s,t)\in\mathbb{R}^{+}\times[0,T]$;
    \label{B2}
    \item $\sigma(s,t)\in C^{2,1}_{\alpha}(\mathbb{R}^{+}\times[0,T])$;
    \label{B3}
    \item $\gamma(s,t,z)\in C^{2,1}_{\alpha}(\mathbb{R}^{+}\times[0,T])$ with the H\"older continuity being uniform in $z$;
    \label{B5}
	\item $|\omega(s)|\leq C$ for all $s\in\mathbb{R}^{+}$;
	\label{B6}
    \item $\omega(s)\in C^{2}_{\alpha}(\mathbb{R}^{+})$;
    \label{B7}
	\item $g(s)$ is Lipschitz continuous;
	\label{B8}
    \item $g(s)\in C^4_{\alpha}(\mathbb{R}^{+})$.
    \label{B9}
\end{enumerate}
}

\medskip

{\bf Assumptions (C)}\\
{\it There exist a constant $C > 0$ such that
\begin{enumerate}[label=(C\arabic*)]
    \item $|\frac{\partial\mu(s,t)}{\partial t}|\leq Cs$,
    $|\frac{\partial^2\mu(s,t)}{\partial s^2}|\leq \frac{C}{s}$ for all $(s,t)\in\mathbb{R}^{+}\times[0,T]$;
    \label{C1}
	\item $|\frac{\partial\sigma(s,t)}{\partial t}|\leq Cs$, $|\frac{\partial^2\sigma(s,t)}{\partial s^2}|\leq \frac{C}{s}$ for all $(s,t)\in\mathbb{R}^{+}\times[0,T]$;
	\label{C2}
    \item $|\frac{\partial\gamma(s,t,z)}{\partial t}|\leq Cs$,
    $|\frac{\partial^2\gamma(s,t,z)}{\partial s^2}|\leq\frac{C}{s}$ for all $(s,t,z)\in\mathbb{R}^{+}\times[0,T]\times\mathbb{R}$;
   	\label{C3}
	\item $|\frac{d\omega(s)}{ds}|\leq \frac{C}{s}$, $|\frac{d^2\omega(s)}{ds^2}|\leq \frac{C}{s^2}$ for all $s\in\mathbb{R}^{+}$;
	\label{C4}
    \item $g(s)\in C^3_{\text{pol}}(\mathbb{R}^{+})$.
   	\label{C5}
\end{enumerate}
}
\begin{theorem}\label{TheoremAboutConvexityOfValueFunctionDiffusionWithJumps}
Let all assumptions of Theorem \ref{mainTheoremAmericanOptionConvexity} be satisfied.
We assume additionally that conditions (B) and (C) 
hold true.
Then $V^{\omega}_{\text{\rm E}}(s, t)$ is convex with respect to $s$ at all times
$t\in[0,T]$.
\end{theorem}

\begin{proof}
\label{ProofTheoremWithJumps}
The first part of the proof proceeds in a similar way as the proof of \cite[Prop. 4.1]{propertiesOfOptionPrices}.

Let \begin{equation*}
\mathcal{L}V^{\omega}_{\text{\rm E}}(s, t) = -\frac{\partial V^{\omega}_{\text{\rm E}}(s, t)}{\partial t} - A_t^C V^{\omega}_{\text{\rm E}}(s, t) - A^J_t V^{\omega}_{\text{\rm E}}(s, t) +\omega(s)V^{\omega}_{\text{\rm E}}(s, t),
\end{equation*}
where $A_t$ is the linear second-order differential operator of the form
\begin{equation*}
A_t^C V^{\omega}_{\text{\rm E}}(s, t) = \beta(s, t) \frac{\partial^2 V^{\omega}_{\text{\rm E}}(s, t)}{\partial s^2} + \mu(s,t) \frac{\partial V^{\omega}_{\text{\rm E}}(s, t)}{\partial s}
\end{equation*}
with $\beta(s, t) = \frac{\sigma^2(s, t)}{2}$
and $A^J_t$ is the integro-differential operator given by
\begin{equation*}
A^J_t V^{\omega}_{\text{\rm E}}(s, t) = \int_{\mathbb{R}} \left(V^{\omega}_{\text{\rm E}}(s+\gamma(s,t,z), t) - V^{\omega}_{\text{\rm E}}(s,t) - \gamma(s,t,z)\frac{\partial V^{\omega}_{\text{\rm E}}(s,t)}{\partial s}\right) m(dz).
\end{equation*}

\begin{lemma}\label{TheoremAboutExistenceOfValueFunctionDiffusionWithJumps}
Let Assumptions (A) and (B) 
hold and assume that the stock price process $S_t$ follows \eqref{diffusionWithJumps}.
Then $V^{\omega}_{\text{\rm E}}(s,t)\in C^{4,1}_{\alpha}(\mathbb{R}^{+}\times[0,T])\cap C_{\text{pol}}(\mathbb{R}^{+}\times[0,T])$
and it is the solution to the Cauchy problem
\begin{equation}\label{CauchyProblem}
  \begin{cases}
    \mathcal{L} V^{\omega}_{\text{\rm E}}(s, t) = 0, \quad & (s, t)\in\mathbb{R}^{+}\times[0,T), \\
	V^{\omega}_{\text{\rm E}}(s, T) = g(s), \quad & s\in\mathbb{R}^{+}.
  \end{cases}
\end{equation}
\end{lemma}


\begin{lemma}\label{LemmaAboutBoundednessOfSecondDerivativeDiffusionWithJumps}
Let Assumptions (A), (B) and (C) 
hold and assume that the stock price process $S_t$ follows \eqref{diffusionWithJumps}.
Then there exist constants $n>0$ and $K>0$ such that the value function $V^{\omega}_{\text{\rm E}}(s, t)$ satisfies
\begin{equation*}
\left|\frac{\partial^2 V^{\omega}_{\text{\rm E}}(s, t)}{\partial s^2}\right| \leq K(s^{-n} + s^{n})
\end{equation*}
for all $(s, t) \in \mathbb{R}^{+}\times[0,T]$.
\end{lemma}

Proofs of both above lemmas are given in Appendix.

We introduce the function $u^{\omega}:\mathbb{R}^{+}\times[0,T]\rightarrow \mathbb{R}^{+}$ of the form
\begin{equation*}
u^{\omega}(s,t) := V^{\omega}_{\text{\rm E}}(s, T-t)
\end{equation*}
and we prove convexity of $u^{\omega}(s,t)$ with respect to $s$. Note that it is equivalent to the convexity of the value function $V^{\omega}_{\text{\rm E}}(s, t)$ in $s$.
Furthermore, based on Lemma \ref{TheoremAboutExistenceOfValueFunctionDiffusionWithJumps}, the function $u^{\omega}(s,t)$ solves the Cauchy problem of the form
\begin{equation*}\label{CauchyProblem3}
  \begin{cases}
     \frac{\partial u^{\omega}(s,t)}{\partial t}= \hat{\mathcal{L}} u^{\omega}(s,t), \quad & (s, t)\in\mathbb{R}^{+}\times(0,T], \\
	u^{\omega}(s, 0) = g(s), \quad & s\in\mathbb{R}^{+},
  \end{cases}
\end{equation*}
where
\begin{align*}
\hat{\mathcal{L}} u^{\omega}(s,t) &= \beta(s, t) \frac{\partial^2 u^{\omega}(s,t)}{\partial s^2} + \mu(s,t) \frac{\partial u^{\omega}(s,t)}{\partial s} - \omega(s)u^{\omega}(s,t) \\ &+ \int_{\mathbb{R}} \left(u^{\omega}(s+\gamma(s,t,z), t) - u^{\omega}(s,t) - \gamma(s,t,z)\frac{\partial u^{\omega}(s,t)}{\partial s}\right) m(dz)
\end{align*}
with $\beta(s, t) = \frac{\sigma^2(s, t)}{2}$.
Observe that by Lemma \ref{LemmaAboutBoundednessOfSecondDerivativeDiffusionWithJumps} there exist constants $n>0$ and $K>0$ such that
\begin{equation}\label{inequalityFromLemmaWithJumps}
\left|\frac{\partial^2 u^{\omega}(s, t)}{\partial s^2}\right| \leq K (s^{-n} + s^n)
\end{equation}
for all $(s, t)\in\mathbb{R}^{+}\times[0,T]$.

Let us now define a convex function $\kappa: \mathbb{R}^{+} \rightarrow \mathbb{R}^{+}$ of the form
\begin{equation*}
\kappa(s):= s^{n+3} + s^{-n+1}
\end{equation*}
with
\begin{equation*}
\frac{d^2 \kappa(s)}{ds^2} = (n + 3)(n + 2)s^{n+1} + n(n - 1)s^{-n-1}
\end{equation*}
and
\begin{align*}
\frac{d^2(\hat{\mathcal{L}}\kappa(s))}{ds^2} &=
\frac{\partial^2\beta(s, t)}{\partial s^2}\frac{d^2 \kappa(s)}{ds^2} + 2\frac{\partial\beta(s, t)}{\partial s} \frac{d^3 \kappa(s)}{ds^3}
+ \beta(s, t) \frac{d^4 \kappa(s)}{ds^4} \\
&+ \frac{\partial^2 \mu(s,t)}{\partial s^2} \frac{d \kappa(s)}{ds} + 2\frac{\partial\mu(s,t)}{\partial s} \frac{d^2 \kappa(s)}{ds^2} + \mu(s, t)\frac{d^3 \kappa(s)}{ds^3} \\
&- \frac{d^2\omega(s)}{ds^2}\kappa(s) - 2\frac{d\omega(s)}{ds}\frac{d \kappa(s)}{ds} - \omega(s)\frac{d^2 \kappa(s)}{ds^2} \\
&+ \int_{\mathbb{R}} \Bigg(\frac{d^2 \kappa(s+\gamma(s, t, z))}{ds^2}\left(1+\frac{\partial\gamma(s,t,z)}{\partial s}\right)^2 \\ &+ \frac{d\kappa(s+\gamma(s,t,z))}{ds}\frac{\partial^2 \gamma(s,t,z)}{\partial s^2} - \gamma(s,t,z) \frac{d^3 \kappa(s)}{ds^3} \\
&- \left(1+2\frac{\partial \gamma(s,t,z)}{\partial s}\right)\frac{d^2\kappa(s)}{ds^2} - \frac{\partial^2\gamma(s,t,z)}{\partial s^2}\frac{d\kappa(s)}{ds}\Bigg)m(dz).
\end{align*}
The assumptions that we put on the coefficients $\mu$, $\sigma$ and $\gamma$ and function $\omega$ and their derivatives imply that each component of the above expression grows at most like $s^{n+1}$ for large $s$ and like $s^{-n-1}$ for small $s$.
The same behaviour characterises
$\frac{d^2\kappa(s)}{ds^2}$.

In addition, we define the function $\vartheta:\mathbb{R}^{+}\times[0,T]\rightarrow\mathbb{R}$ given by
\begin{equation*}
\vartheta(s, t) := \left(\frac{\partial^2\mu(s, t)}{\partial s^2} - 2 \frac{d\omega(s)}{ds}\right)\frac{d\kappa(s)}{ds}
\end{equation*}
which also behaves like $\frac{d^2(\hat{\mathcal{L}}\kappa(s))}{ds^2}$ at $+\infty$ and $-\infty$.

Hence we claim that there exist a positive constant $C$ such that
\begin{equation}\label{dependenceWithConstantC}
C\frac{d^2\kappa(s)}{ds^2} - \frac{d^2(\hat{\mathcal{L}}\kappa(s))}{ds^2} >
-\vartheta(s, t).
\end{equation}
In the second part of the proof, we define the auxiliary function
\begin{equation}\label{Vepsilon}
u^{\omega}_{\varepsilon}(s, t) := u^{\omega}(s, t) + \varepsilon e^{Ct} \kappa(s)
\end{equation}
for some $\varepsilon>0$.

We carry out a proof by contradiction. Let us then assume that
$u^{\omega}_{\varepsilon}(s,t)$ is not convex. For this purpose, we denote by $\Lambda$ the set of points for which $u^{\omega}_{\varepsilon}(s, t)$ is not convex, i.e.
\begin{equation*}
\Lambda := \{(s, t)\in\mathbb{R}^{+}\times[0,T]:\frac{\partial^2 u^{\omega}_{\varepsilon}(s, t)}{\partial s^2}<0\}
\end{equation*}
and we assume that the set $\Lambda$ is not empty.

From Lemma \ref{LemmaAboutBoundednessOfSecondDerivativeDiffusionWithJumps} we know that $u^{\omega}(s, t)$ satisfies \eqref{inequalityFromLemmaWithJumps}. Due to this fact and using \eqref{Vepsilon}
we claim that there exist a positive constant $R$ such that $\Lambda \subseteq [R^{-1}, R] \times [0,T]$. This is a direct consequence of such a choice of $u^{\omega}_{\varepsilon}(s, t)$ in \eqref{Vepsilon} so that $\frac{d^2 \kappa(s)}{ds^2}$ grows faster than $\frac{\partial^2 u^{\omega}(s, t)}{\partial s^2}$ for both large and small values of $s$.

Consequently, the set $\Lambda$ is a bounded set.
Since the closure of a bounded set is also bounded, we conclude that the closure of $\Lambda$, i.e. $\text{cl}(\Lambda)$, is compact.

Due to the fact that a compact set always contains its infimum we can define
\begin{equation*}
t_0 := \inf\{t\geq 0: (s, t) \in \text{cl}(\Lambda)\text{ for some } s \in \mathbb{R^{+}}\}.
\end{equation*}
From the initial condition, i.e. $u^{\omega}(s, 0) = g(s)$ and convexity of $g$ we have
\begin{equation*}
\frac{d^2 u^{\omega}_{\varepsilon}(s, 0)}{d s^2} = \frac{d^2 (g(s)+\varepsilon \kappa(s))}{d s^2}\geq \varepsilon \frac{d^2 \kappa(s)}{ds^2}>0
\end{equation*}
for all $s\in\mathbb{R}^{+}$.
Hence we can conclude that $t_0>0$.

Moreover, at the point when the infimum is attained, i.e. $(s_0, t_0)$ for some $s_0\in\mathbb{R}^{+}$
\begin{equation*}
\frac{\partial^2 u^{\omega}_{\varepsilon}(s_0, t_0)}{\partial s^2} = 0.
\end{equation*}
This is a consequence of the continuity of the function $\frac{\partial^2 u^{\omega}_{\varepsilon}(s, t)}{\partial s^2}$ in $s$.
In addition, for $t\in[0, t_0)$ we have $\frac{\partial^2 u^{\omega}_{\varepsilon}(s_0, t)}{\partial s^2}> 0$ and thus, by applying the symmetry of second derivatives at $t= t_0$, we derive
\begin{equation}\label{onePartOfFinalEquation}
\frac{\partial^2}{\partial s^2} \left(\frac{\partial u^{\omega}_{\varepsilon}(s_0, t_0)}{\partial t}\right) =
\frac{\partial}{\partial t}\left(\frac{\partial^2 u^{\omega}_{\varepsilon}(s_0, t_0)}{\partial s^2}\right) \leq 0.
\end{equation}
Furthermore, at $(s_0, t_0)$ we also have
\begin{align*}
\frac{\partial^2 (\hat{\mathcal{L}}u^{\omega}_{\varepsilon}(s_0, t_0))}{\partial s^2} &= \frac{\partial^2\beta(s_0, t_0)}{\partial s^2}\frac{\partial^2 u^{\omega}_{\varepsilon}(s_0, t_0)}{\partial s^2} + 2\frac{\partial\beta(s_0, t_0)}{\partial s}\frac{\partial^3 u^{\omega}_{\varepsilon}(s_0, t_0)}{\partial s^3} \nonumber \\
&+ \beta(s_0, t_0) \frac{\partial^4 u^{\omega}_{\varepsilon}(s_0, t_0)}{\partial s^4} + \frac{\partial^2 \mu(s_0, t_0)}{\partial s^2} \frac{\partial u^{\omega}_{\varepsilon}(s_0, t_0)}{\partial s} \nonumber \\
&+ 2\frac{\partial\mu(s_0,t_0)}{\partial s} \frac{\partial^2 u^{\omega}_{\varepsilon}(s_0, t_0)}{\partial s^2} + \mu(s_0, t_0)\frac{\partial^3 u^{\omega}_{\varepsilon}(s_0, t_0)}{\partial s^3} \nonumber \\
&- \frac{d^2\omega(s_0)}{ds^2}u^{\omega}_{\varepsilon}(s_0, t_0) - 2\frac{d\omega(s_0)}{ds}\frac{\partial u^{\omega}_{\varepsilon}(s_0, t_0)}{\partial s} - \omega(s_0)\frac{\partial^2 u^{\omega}_{\varepsilon}(s_0, t_0)}{\partial s^2}\\
&+ \int_{\mathbb{R}} \Bigg(\frac{\partial^2 u^{\omega}_{\varepsilon}(s_0+\gamma(s_0, t_0, z), t_0)}{\partial s^2}\left(1+\frac{\partial\gamma(s_0,t_0,z)}{\partial s}\right)^2 \nonumber \\
&+ \frac{\partial u^{\omega}_{\varepsilon}(s_0+\gamma(s_0,t_0,z), t_0)}{\partial s}\frac{\partial^2 \gamma(s_0,t_0,z)}{\partial s^2} - \gamma(s_0,t_0,z) \frac{\partial^3 u^{\omega}_{\varepsilon}(s_0, t_0)}{\partial s^3} \nonumber \\
&- \left(1+2\frac{\partial\gamma(s_0,t_0,z)}{\partial s}\right)\frac{\partial^2 u^{\omega}_{\varepsilon}(s_0, t_0)}{\partial s^2} - \frac{\partial^2\gamma(s_0,t_0,z)}{\partial s^2}\frac{\partial u^{\omega}_{\varepsilon}(s_0, t_0)}{\partial s}\Bigg)m(dz). \nonumber
\end{align*}
Since $\frac{\partial^2 u^{\omega}_{\varepsilon}(s_0, t_0)}{\partial s^2} = 0$ and $\frac{\partial^2 u^{\omega}_{\varepsilon}(s, t_0)}{\partial s^2}$ has a local minimum at $s=s_0$, we have $\frac{\partial^3 u^{\omega}_{\varepsilon}(s_0, t_0)}{\partial s^3} = 0$ and \linebreak $\frac{\partial^4 u^{\omega}_{\varepsilon}(s_0, t_0)}{\partial s^4} \geq 0$.
Thus
\begin{equation*}\label{secDerivOfOperatorLj}
\begin{aligned}
\frac{\partial^2(\hat{\mathcal{L}}u^{\omega}_{\varepsilon}(s_0, t_0))}{\partial s^2} &\geq \frac{\partial^2 \mu(s_0, t_0)}{\partial s^2} \frac{\partial u^{\omega}_{\varepsilon}(s_0, t_0)}{\partial s} - \frac{d^2\omega(s_0)}{ds^2}u^{\omega}_{\varepsilon}(s_0, t_0) \\ &- 2\frac{d\omega(s_0)}{ds}\frac{\partial u^{\omega}_{\varepsilon}(s_0, t_0)}{\partial s} \\ &+ \int_{\mathbb{R}} \Bigg(\frac{\partial u^{\omega}_{\varepsilon}(s_0+\gamma(s_0,t_0,z), t_0)}{\partial s}\frac{\partial^2 \gamma(s_0,t_0,z)}{\partial s^2} \\ &- \frac{\partial u^{\omega}_{\varepsilon}(s_0, t_0)}{\partial s} \frac{\partial^2\gamma(s_0,t_0,z)}{\partial s^2}\Bigg)m(dz).
\end{aligned}
\end{equation*}
Since $u^{\omega}_{\varepsilon}(s, t_0)$ is convex in $s$ and
$\frac{\partial^2 u^{\omega}_{\varepsilon}(s_0, t_0)}{\partial s^2} = 0$,
applying \eqref{inequalityOnGamma2} we can conclude that the integral part of the above expression is nonnegative.
Moreover, the concavity of $\omega$ and \eqref{inequalityOnMuOmegaG2} imply that
\begin{equation}\label{secondPartOfFinalEquation}
\frac{\partial^2(\hat{\mathcal{L}}u^{\omega}_{\varepsilon}(s_0, t_0))}{\partial s^2} \geq \varepsilon e^{C t_0} \left(\frac{\partial^2 \mu(s_0, t_0)}{\partial s^2} - \frac{d^2\omega(s_0)}{ds^2}\right)\frac{d\kappa(s_0)}{ds} = \varepsilon e^{C t_0} \vartheta(s_0, t_0).
\end{equation}
Combining \eqref{dependenceWithConstantC} with \eqref{onePartOfFinalEquation} and \eqref{secondPartOfFinalEquation} at
$(s_0, t_0)$ we derive that
\begin{equation*}
\begin{aligned}
\frac{\partial^2}{\partial s^2}\left(\frac{\partial u^{\omega}_{\varepsilon}(s_0, t_0)}{\partial t} - \hat{\mathcal{L}}u^{\omega}_{\varepsilon}(s_0, t_0)\right) &= \varepsilon e^{Ct_0}\frac{d^2}{d s^2}(C\kappa(s_0) - \hat{\mathcal{L}}\kappa(s_0)) \\
&> -\varepsilon e^{Ct_0}\vartheta(s_0, t_0) \geq \frac{\partial^2}{\partial s^2}\left(\frac{\partial u^{\omega}_{\varepsilon}(s_0, t_0)}{\partial t} - \hat{\mathcal{L}}u^{\omega}_{\varepsilon}(s_0, t_0)\right)
\end{aligned}
\end{equation*}
which is a contradiction. This confirms that the set $\Lambda$ is empty, and thus $u^{\omega}_{\varepsilon}(s, t)$ is a convex function.
Finally, letting $\varepsilon\rightarrow 0$ we conclude that $u^{\omega}(s, t)$ is convex in $s$ for all $t\in[0,T]$.
\end{proof}

Using the same arguments like in the proof of \cite[Thm. 4.1]{propertiesOfOptionPrices}, we can resign from Assumptions (B) and (C) 
in Theorem \ref{TheoremAboutConvexityOfValueFunctionDiffusionWithJumps}, that is, the following theorem holds true.
\begin{theorem}\label{TheoremAboutConvexityOfValueFunctionDiffusionWithJumps2}
Let assumptions of Theorem \ref{mainTheoremAmericanOptionConvexity} hold true.
Then $V^{\omega}_{\text{\rm E}}(s, t)$ is convex with respect to $s$ at all times
$t\in[0,T]$.
\end{theorem}

We are ready to give the proof of our first main result.

\begin{proof}[\textbf{Proof of Theorem \ref{mainTheoremAmericanOptionConvexity}}]
As noted in \cite[Sec. 7]{propertiesOfOptionPrices}, under conditions \ref{A1}--\ref{A4}, for each $p\geq 1$ there exists a constant $C$
such that the stock price process given in \eqref{diffusionWithJumps} satisfies
\begin{equation*}\label{inequalityFromGihman}
\mathbb{E}_s\left[\sup_{0\leq t\leq T}|S_t|^p\right]\leq C(1+s^p).
\end{equation*}
Together with \ref{A5} and \ref{A6} it implies that the value function given by
\begin{equation*}
V^{\omega}_{\text{\rm A}_T}(s, t) := \sup_{\tau\in\mathcal{T}^T_t}
\mathbb{E}_{s, t}[e^{-\int_t^{\tau} \omega(S_w) dw} g(S_\tau)]
\end{equation*}
is well-defined, where $\mathcal{T}^T_t$ is the family of
$\mathcal{F}_t$-stopping times with values in $[t, T]$ for fixed maturity $T>0$.
Moreover, we denote
\begin{equation*}
V^{\omega}_{\text{\rm A}_T}(s) := V^{\omega}_{\text{\rm A}_T}(s, 0).
\end{equation*}
Let us define now a Bermudan option with the value function of the form
\begin{equation*}\label{berdmudanOption}
V^{\omega}_{\text{B}_\Xi}(s, t) := \sup_{\tau\in\mathcal{T}_{\Xi}}
\mathbb{E}_{s, t}[e^{-\int_t^\tau \omega(S_w) dw} g(S_\tau)],
\end{equation*}
where $\mathcal{T}_{\Xi}$ is the set of stopping times with values in
\begin{equation*}\label{BNSet}
\text{B}_\Xi = \left\{\frac{n}{2^\Xi}(T-t)+t: n = 0, 1, ..., 2^\Xi\right\},
\end{equation*}
where $\Xi$ is some positive integer number.
To simplify the notation, we denote
\begin{equation*}
V^{\omega}_{\text{B}_\Xi}(s) := V^{\omega}_{\text{B}_\Xi}(s, 0).
\end{equation*}
In contrast to the American options, the Bermudan options are the options that can be exercised at one of the finitely many number of times.

Now we show that $V^{\omega}_{\text{B}_\Xi}(s, t)$ inherits the property of convexity from its European equivalent $V^{\omega}_{\text{\rm E}}(s, t)$.
Next, we generalise this result to the American case $V^{\omega}_{\text{\rm A}}(s)$.

\begin{lemma}\label{TheoremAboutConvexityOfBermudan}
Let assumptions of Theorem \ref{mainTheoremAmericanOptionConvexity} hold true.
Then $V^{\omega}_{\text{B}_\Xi}(s, t)$ is convex with respect to $s$ at all times $t\in[0,T]$.
\end{lemma}
Its proof is given in Appendix.

As the possible exercise times of the Bermudan option get denser, the value function $V^{\omega}_{\text{B}_\Xi}(s, t)$ converges to $V^{\omega}_{\text{\rm A}_T}(s, t)$. To formalise this result, we proceed as follows. For a given stopping time $\tau_0^T$ that takes values in $[0,T]$, we define
\begin{equation*}
\tau_\Xi := \inf\{t\in \text{B}_\Xi: t\geq \tau_0^T\}.
\end{equation*}
Then $\tau_\Xi \in \text{B}_\Xi$ is a stopping time and $\tau_\Xi\rightarrow\tau_0^T$ almost surely as $\Xi\rightarrow+\infty$.
Moreover, by the dominated convergence theorem, we obtain
\begin{equation*}
\begin{aligned}
&\left|\mathbb{E}_{s}\left[e^{-\int_{0}^{\tau_\Xi}\omega(S_w) dw}g(S_{\tau_\Xi})\right] - \mathbb{E}_{s}\left[e^{-\int_{0}^{\tau_0^T}\omega(S_w) dw}g(S_{\tau_0^T})\right]\right| \\ &\leq
\mathbb{E}_{s}\left|e^{-\int_{0}^{\tau_\Xi}\omega(S_w) dw} g(S_{\tau_\Xi}) - e^{-\int_{0}^{\tau_0^T}\omega(S_w) dw} g(S_{\tau_0^T})\right|\rightarrow 0
\end{aligned}
\end{equation*}
as $\Xi\rightarrow\infty$.
Therefore, it follows that
\begin{equation*}
\liminf_{\Xi\rightarrow\infty} V^{\omega}_{\text{B}_\Xi}(s)\geq  V^{\omega}_{\text{\rm A}_T}(s)
\end{equation*}
Since it is obvious that
\begin{equation*}
V^{\omega}_{\text{B}_\Xi}(s)\leq  V^{\omega}_{\text{\rm A}_T}(s),
\end{equation*}
we finally derive
\begin{equation*}
V^{\omega}_{\text{B}_\Xi}(s)\rightarrow  V^{\omega}_{\text{\rm A}_T}(s)
\end{equation*}
as $\Xi\rightarrow\infty$.
To receive our claim we take the maturity $T$ tending to infinity.
\end{proof}

\begin{proof}[\textbf{Proof of Theorem \ref{mainTheoremGeometricLevy}}]

Before we proceed to the actual proof let us remind the main exit identities from \cite{LiBo}.
Let
\begin{equation*}
\sigma^{+}_a:=\inf\{t> 0: X_t\geq a\}, \qquad \sigma^{-}_a:=\inf\{t> 0: X_t\leq a\}
\end{equation*}
for some $a\in\mathbb{R}$.
Then for the function $\eta$ defined in \eqref{eta2} we have
\begin{align}
&\mathbb{E}\left[e^{-\int_0^{\sigma_{a}^{+}} \eta(X_w)\,dw}; \sigma_{a}^{+}<\infty\mid X_0=x\right] =\frac{\mathcal{H}^{(\eta)}(x)}{\mathcal{H}^{(\eta)}(a)},\label{upwardexit}\\
&\mathbb{E}\left[e^{-\int_0^{\sigma_0^{-}} \eta(X_w)\,dw}; \sigma_0^{-}<\infty\mid X_0=x\right] = \mathcal{Z}^{(\eta)}(x) - c_{\mathcal{Z}^{(\eta)}/\mathcal{W}^{(\eta)}} \mathcal{W}^{(\eta)}(x),\label{dwodexit}
\end{align}
where $c_{\mathcal{Z}^{(\eta)}/\mathcal{W}^{(\eta)}} = \lim_{z\rightarrow\infty} \frac{\mathcal{Z}^{(\eta)}(z)}{\mathcal{W}^{(\eta)}(z)}$
and we use condition $\eta(x)=c$ for all $x\leq 0$ and some constant $c\in\mathbb{R}$ in the first identity.
Denoting
\begin{equation*}
\tau^{+}_a:=\inf\{t> 0: S_t\geq a\}, \qquad \tau^{-}_a:=\inf\{t> 0: S_t\leq a\}
\end{equation*}
and keeping in mind that $S_t=e^{X_t}$, from \eqref{upwardexit} and \eqref{dwodexit} we can conclude that
\begin{align}
&\mathbb{E}_{s}\left[e^{-\int_0^{\tau_{a}^{+}} \omega(S_w)\,dw}; \tau_{a}^{+}<\infty\right] =\frac{\mathscr{H}^{(\omega)}(s)}{\mathscr{H}^{(\omega)}(a)},\label{upwardexit2}\\
&\mathbb{E}_{s}\left[e^{-\int_0^{\tau_1^{-}} \omega(S_w)\,dw}; \tau_1^{-}<\infty\right] = \mathscr{Z}^{(\omega)}(s) - c_{\mathscr{Z}^{(\omega)}/\mathscr{W}^{(\omega)}} \mathscr{W}^{(\omega)}(s), \nonumber
\end{align}
where $\omega(s) = \omega(e^x) = \eta(x)$ and the functions $\mathscr{Z}^{(\omega)}(s)$, $\mathscr{W}^{(\omega)}(s)$, $\mathscr{H}^{(\omega)}(s)$ were defined in \eqref{scaleFunctionW2}, \eqref{scaleFunctionZ2} and \eqref{scaleH2}.

We consider three possible cases of a position of the initial state $S_0=s$ of the process $S_t$.
\begin{enumerate}
\item $s<l$:
As the process $S_t$ is spectrally negative and starts below the interval $[l, u]$, it can enter this interval only in a continuous way and hence $\tau_{l, u} = \tau^{+}_l$ and $S_{\tau_{l, u}} = l$. Thus from \eqref{upwardexit2}
\begin{equation*}\label{ResultLiBoTh2.5.}
\begin{aligned}
v^{\omega}_{\text{\rm A}^{\text{\rm Put}}}(s, l, u) &= \mathbb{E}_{s}\left[e^{-\int_0^{\tau^{+}_{l}}\omega(S_w)dw}; S_{\tau^{+}_{l}} = l\right](K-l) \\ &=
\frac{\mathscr{H}^{(\omega)}(s)}{\mathscr{H}^{(\omega)}(l)}(K-l).
\end{aligned}
\end{equation*}
\item $s\in[l, u]$:
If the process $S_t$ starts inside the interval $[l, u]$ which is an optimal stopping region, we decide to exercise our option immediately, i.e. $\tau_{l, u} = 0$. Therefore, we have
\begin{equation*}
v^{\omega}_{\text{\rm A}^{\text{\rm Put}}}(s, l, u) = K-s.
\end{equation*}
\item $s>u$:
There are three possible cases of entering the interval $[l, u]$ by the process $S_t$ when it starts above $u$: either $S_t$ enters $[l, u]$ continuously going downward or it jumps from $(u, +\infty)$ to $(l, u)$ or
$S_t$ jumps from the interval $(u, +\infty)$ to the interval $(0, l)$ and then, later, enters $[l, u]$ continuously.
\end{enumerate}
We can distinguish these cases in the following way
\begin{equation}\label{valueFunctionTwoComponents}
\begin{aligned}
v^{\omega}_{\text{\rm A}^{\text{\rm Put}}}(s, l, u) &
=\mathbb{E}_{s}\left[e^{-\int_0^{\tau_{l, u}}\omega(S_w)dw}(K-S_{\tau_{l, u}}); \tau^{-}_{u}<\tau^{-}_{l}\right] \\ &\quad+
\mathbb{E}_{s}\left[e^{-\int_0^{\tau_{l, u}}\omega(S_w)dw} (K-S_{\tau_{l, u}}); \tau^{-}_{u} = \tau^{-}_{l}\right].
\end{aligned}
\end{equation}
To analyse the first component in \eqref{valueFunctionTwoComponents}, note that
\begin{equation*}
\begin{aligned}
&\mathbb{E}_{s}\left[e^{-\int_0^{\tau_{l, u}}\omega(S_w)dw} (K-S_{\tau_{l, u}}); \tau^{-}_{u}<\tau^{-}_{l}\right]
\\ &=
\mathbb{E}_{s}\left[e^{-\int_0^{\tau^{-}_{u}}\omega(S_w)dw}(K-S_{\tau^{-}_{u}}); S_{\tau^{-}_{u}}\in[l, u]\right]
\\ &=
\int_{(l, u)}(K-z)\mathbb{E}_{s}\left[e^{-\int_0^{\tau^{-}_{u}}\omega(S_w)dw}; S_{\tau^{-}_{u}}\in d z\right] \\ &\quad+
(K-u) \mathbb{E}_{s}\left[e^{-\int_0^{\tau^{-}_{u}}\omega(S_w)dw}; S_{\tau^{-}_{u}}=u\right].
\end{aligned}
\end{equation*}
We express now above formulas in terms of $X_t=\log S_t$ process.
Let $x=\log s$ and we recall that in \eqref{eta2} we introduced
functions $\eta(x) = \omega\circ{\rm exp}(x)=\omega(e^x)$ and $\eta_u(x)=\eta(x+\log u).$
Then
\begin{equation}\label{veryLongFormula}
\begin{aligned}
&\mathbb{E}_{s}\left[e^{-\int_0^{\tau_{l, u}}\omega(S_w)dw} (K-S_{\tau_{l, u}}); \tau^{-}_{u}<\tau^{-}_{l}\right]
\\&=
\int_{(\log l, \log u)}(K-e^z)\mathbb{E}\left[e^{-\int_0^{\sigma^{-}_{\log u}}\eta(X_w)dw}; X_{\sigma^{-}_{\log u}}\in d z\mid X_0=x\right] \\ &\quad+
(K-u) \mathbb{E}\left[e^{-\int_0^{\sigma^{-}_{\log u}}\eta(X_w)dw}; X_{\sigma^{-}_{\log u}}=\log u\mid X_0=x\right]
\\ &=
\int_{(0, \log u-\log l)}(K-e^{\log u-y})\mathbb{E}\left[e^{-\int_0^{\sigma^{-}_0}\eta_u(X_w)dw}; -X_{\sigma^{-}_0}\in dy\mid X_0=x-\log u\right] \\ &\quad+
(K-u) \mathbb{E}\left[e^{-\int_0^{\sigma^{-}_0}\eta_u(X_w)dw}; X_{\sigma^{-}_0}=0\mid X_0=x-\log u\right].
\end{aligned}
\end{equation}
From the compensation formula for L\'evy processes given in \cite[Thm. 4.4]{kyprianou}
we have
\begin{equation}\label{formulaWithResolvent}
\mathbb{E}\left[e^{-\int_0^{\sigma^{-}_0}\eta_u(X_w)dw}; -X_{\sigma^{-}_0}\in dy\mid X_0=x-\log u\right]
= \int_0^{\infty}r^{(\eta_u)}(x-\log u, z) \Pi(-z-dy)dz,
\end{equation}
where $r^{(\eta_u)}(x-\log u, z)$ is the resolvent density of $X_t$ killed by potential $\eta_u$ and on exiting from positive
half-line which is, by \cite[Thm. 2.2]{LiBo}, given by
\begin{equation*}\label{res0}
r^{(\eta_u)}(x-\log u, z)=\mathcal{W}^{(\eta_u)}(x-\log u)
\lim_{y\rightarrow\infty}\frac{\mathcal{W}^{(\eta_u)}(y, z)}{\mathcal{W}^{(\eta_u)}(y)}
- \mathcal{W}^{(\eta_u)}(x-\log u, z).\end{equation*}
Note that $r^{(\eta_u)}(\log s-\log u, z)=r(s,u,z)$ for $r(s,u,z)$
given in \eqref{resovent}.

To find $\mathbb{E}\left[e^{-\int_0^{\sigma^{-}_0}\eta_u(X_w)dw}; X_{\sigma^{-}_0}=0\mid X_0=x-\log u\right]$, we consider
\begin{equation*}\label{valueFunctionContinuousPath}
\begin{aligned}
\mathbb{E}\left[e^{-\int_0^{\sigma^{-}_0}\eta_u(X_w)dw + \alpha X_{\sigma^{-}_0}}; \sigma_0^{-}<\infty\mid X_0=x-\log u\right]
\end{aligned}
\end{equation*}
for some $\alpha>0$. Note that using the change of measure given in \eqref{changemeas} it is equal to
\begin{equation}\label{using}
e^{\alpha (x-\log u)}\mathbb{E}^{(\alpha)}\left[e^{-\int_0^{\sigma^{-}_0}\eta_u^\alpha(X_w)dw}; \sigma_0^{-}<\infty\mid X_0=x-\log u\right]
\end{equation}
where $\mathbb{E}^{(\alpha)}$ is the expectation with respect to $\mathbb{P}^{(\alpha)}$ and $\eta_u^\alpha(x) := \eta_u(x)-\psi(\alpha)$.
From \eqref{dwodexit}
we know that
\begin{equation*}
\mathbb{E}^{(\alpha)}\left[e^{-\int_0^{\sigma^{-}_0}\eta_u^\alpha(X_w)dw}; \sigma_0^{-}<\infty\mid X_0=x-\log u\right] = \mathcal{Z}^{(\eta_u^\alpha)}_{\alpha}(x-\log u) - c_{\mathcal{Z}^{(\eta_u^\alpha)}_{\alpha}/\mathcal{W}^{(\eta_u^\alpha)}_{\alpha}}\mathcal{W}^{(\eta_u^\alpha)}_{\alpha}(x-\log u).
\end{equation*}

Moreover, observe that
\begin{equation*}
\begin{aligned}
\mathbb{E}\left[e^{-\int_0^{\sigma^{-}_0}\eta_u(X_w)dw+ \alpha X_{\sigma^{-}_0}}; \sigma_0^{-}<\infty\mid X_0=x-\log u\right] &= \mathbb{E}\left[e^{-\int_0^{\sigma^{-}_0}\eta_u(X_w)dw};X_{\sigma^{-}_0} = 0\mid X_0=x-\log u\right] \\ &+ \mathbb{E}\left[e^{-\int_0^{\sigma^{-}_0}\eta_u(X_w)dw + \alpha X_{\sigma^{-}_0}};X_{\sigma^{-}_0} < 0\mid X_0=x-\log u\right].
\end{aligned}
\end{equation*}

Taking the limit $\alpha\rightarrow\infty$ and using \eqref{using} we derive
\begin{equation}\label{formulaWithLimit}
\begin{aligned}
\lim_{\alpha\rightarrow\infty}e^{\alpha (x-\log u)}
\mathbb{E}^{(\alpha)}\left[e^{-\int_0^{\sigma^{-}_0}\eta_u^\alpha(X_w)dw};\sigma_0^{-}<\infty\mid X_0=x-\log u\right] = \mathbb{E}\left[e^{-\int_0^{\sigma^{-}_0}\eta_u(X_w)dw};X_{\sigma^{-}_0} = 0\mid X_0=x-\log u\right]
\end{aligned}
\end{equation}
and therefore we have
\begin{equation}\label{formulaWithLimit2}
\begin{aligned}
\mathbb{E}\left[e^{-\int_0^{\sigma^{-}_0}\eta_u(X_w)dw};X_{\sigma^{-}_0} = 0\mid X_0=x-\log u\right] = \lim_{\alpha\rightarrow\infty}
e^{\alpha (x-\log u)}\left(\mathcal{Z}^{(\eta^\alpha_u)}_{\alpha}(x-\log u) - c_{\mathcal{Z}^{(\eta^\alpha_u)}_{\alpha}/\mathcal{W}^{(\eta^\alpha_u)}_{\alpha}}\mathcal{W}^{(\eta^\alpha_u)}_{\alpha}(x-\log u)\right).
\end{aligned}
\end{equation}

Furthermore, the second component of \eqref{valueFunctionTwoComponents} is equal to
\begin{equation}\label{veryveryLongFormula}
\begin{aligned}
&\mathbb{E}_{s}\left[e^{-\int_0^{\tau_{l, u}}\omega(S_w)dw} (K-S_{\tau_{l, u}}); \tau^{-}_{u} = \tau^{-}_{l}\right] = \mathbb{E}\left[e^{-\int_0^{\tau_{l, u}}\eta(X_w)dw} (K-e^{X_{\tau_{l, u}}}); \sigma^{-}_{\log u} = \sigma^{-}_{\log l}\mid X_0=x\right]\\ &=
\mathbb{E}\left[e^{-\int_0^{\tau_{l, u}}\eta(X_w)dw} (K-e^{X_{\tau_{l, u}}}); X_{\sigma^{-}_{\log u}}<\log l\mid X_0=x\right] \\ &=
\mathbb{E}\left[e^{-\int_0^{\sigma^{-}_{\log u}}\eta(X_w)dw} \mathbb{E}_{X_{\sigma^{-}_{\log u}}}\left[e^{-\int_0^{\tau_{l, u}}\eta(X_w)dw}(K-e^{X_{\tau_{l, u}}})\right]; X_{\sigma^{-}_{\log u}}<\log l\mid X_0=x\right] \\ &=
\int_{\log u-\log l}^\infty \mathbb{E}\left[e^{-\int_0^{\sigma^{-}_0}\eta_u(X_w)dw} \mathbb{E}\left[e^{-\int_0^{\tau_{l, u}}\eta_u(X_w)dw}(K-e^{X_{\tau_{l, u}}})\mid X_0=\log u-y\right]; -X_{\sigma^{-}_0}\in dy\right] \\ &=
\int_{\log u-\log l}^\infty \frac{\mathcal{H}^{(\eta_u)}(\log u-y)}{\mathcal{H}^{(\eta_u)}(\log l)}(K-l)\mathbb{E}\left[e^{-\int_0^{\sigma^{-}_0} \eta_u(X_w)dw}; -X_{\sigma^{-}_0}\in dy\mid X_0=x-\log u\right].
\end{aligned}
\end{equation}
Now we have to express all scale functions in terms of the $S_t$ scale functions defined in \eqref{scaleFunctionW2}--\eqref{scaleFunctionW2b}
with $x=\log s$ and using \eqref{eta2}.
Finally, using \eqref{valueFunctionTwoComponents} together with \eqref{veryLongFormula}, \eqref{formulaWithResolvent}, \eqref{formulaWithLimit} and \eqref{veryveryLongFormula} completes the proof of the first part of the theorem.
If $l^*=0$ then we can proceed as before except that we do not need identity \eqref{upwardexit2} and hence assumption
\eqref{etaassump} is indeed superfluous.
\end{proof}

\begin{proof}[\textbf{Proof of Theorem \ref{smoothTheorem}}]
From the facts that $V^{\omega}_{\text{\rm A}}(s)\in D(\mathcal{A})$ and that the L\'evy process $X_t$ is right-continuous and left-continuous over stopping times, using
classical arguments,
we can conclude that $V^{\omega}_{\text{\rm A}}(s)$ solves uniquely equation \eqref{CauchyProblem1}; see \cite[Thm. 2.4, p. 37]{peskir} and \cite{peskirTiziano} for details.
More formally, our function as a convex function is continuous (in whole domain).
Since our boundary is sufficiently regular we know that the
Dirichlet/Poisson problem can be solved uniquely
in $D(\mathcal{A})$. This solution can then be identified with the value function $V^{\omega}_{\text{\rm A}}(s)$
itself using the stochastic calculus or infinitesimal generator techniques in the continuation set; see
\cite[p. 131]{peskir} for further details. Similar considerations have been performed only for a local operator
$\mathcal{A}$ in
\cite[Thm. 1]{Strul}. Note that we can handle the non-local case of $\mathcal{A}$ only thanks to proving
the convexity of the value function first.
We are left with the proof of the smoothness at the boundary of stopping set. We prove it at $u^*$. The
proof at lower end $l^*$ follows exactly in the same way.
We choose to follow the idea given in \cite{Lamberton} although one can also apply \cite{peskirTiziano}
or similar arguments as the ones given in \cite{tumilewicz}.

Suppose then that $1$ is regular for $(0, 1)$. Since $V^{\omega}_{\text{\rm A}}(s)\geq g(s)$ and $V^{\omega}_{\text{\rm A}}(u^*)=g(u^*)$,
we have
\[\frac{V^{\omega}_{\text{\rm A}}(u^*+h)-V^{\omega}_{\text{\rm A}}(u^*)}{h}\geq \frac{g(u^*+h)-g(u^*)}{h}.\]
Hence
\[\liminf_{h\downarrow 0}\frac{V^{\omega}_{\text{\rm A}}(u^*+h)-V^{\omega}_{\text{\rm A}}(u^*)}{h}\geq g^\prime(u^*).\]
To get the opposite inequality we introduce
\[\tau_h=\inf\{t \geq 0: S_t\in [l^*,u^*]|S_{0}=u^*+h\}.\]
By assumed regularity, $\tau_h\rightarrow 0$ a.s. as $h\downarrow 0$.
Moreover, by Markov property
\[V^{\omega}_{\text{\rm A}}(u^*)\geq \mathbb{E}_{\log u^*}\left[e^{-\int_0^{\tau_h} \omega (S_w) dw} g\left(S_{\tau_h}
\right)\right]
.\]
Then by \ref{B6} and the space homogeneity of $\log S_t$,
\begin{align*}&\frac{V^{\omega}_{\text{\rm A}}(u^*+h)-V^{\omega}_{\text{\rm A}}(u^*)}{h}\\
&\leq \frac{\mathbb{E}_{u^*+h}\left[e^{-\int_0^{\tau_h} \omega (S_w) dw} g\left(S_{\tau_h}\right)\right]-\mathbb{E}_{u^*}\left[e^{-\int_0^{\tau} \omega (S_w) dw} g\left(S_{\tau}\right)\right]}{h}\\
&\leq \frac{\mathbb{E}_{u^*+h}\left[e^{-\int_0^{\tau_h} \omega (S_w) dw} g\left((u^*+h)S_{\tau_h}\right)\right]-\mathbb{E}_{u^*}\left[e^{-\int_0^{\tau} \omega (S_w) dw} g\left(u^*S_\tau\right)\right]}{h}\end{align*}
and
\begin{align*}
\limsup_{h\downarrow 0}\frac{V^{\omega}_{\text{\rm A}}(u^*+h)-V^{\omega}_{\text{\rm A}}(u^*)}{h}\leq g^\prime(u^*),
\end{align*}
where we use the fact that $g$ is continuously differentiable at $u^*$ in the last step. This completes the proof.
\end{proof}

\begin{proof}[\textbf{Proof of Theorem \ref{putCallSymmetry}}]
We recall that
\begin{align*}
V^{\omega}_{\text{\rm A}^{\text{\rm Call}}}(s, K, \zeta, \sigma, \Pi, l, u) &=
\mathbb{E}_{s}[e^{-\int_0^{\tau_{l, u}}\omega(S_w)dw} (S_{\tau_{l, u}} - K)^{+}] \\
&= \mathbb{E}[e^{-\int_0^{\tau_{l, u}}\eta(X_w)dw} (e^{X_{\tau_{l, u}}} - K)^{+}\mid X_0=x],
\end{align*}
where $x=\log S_0=\log s$.
By our assumption for general L\'evy process $X_t$ we can define a new measure
$\mathbb{P}^{(1)}_1$ via
\begin{equation*}
\left.\frac{d \mathbb{P}^{(1)}_1}{d\mathbb{P}_1}\right\vert_{\mathcal{F}_t} = e^{X_{t} - \psi(1)t};
\end{equation*}
see also \eqref{changemeas} (considered there only for spectrally negative L\'evy process).
Then
\begin{eqnarray}
\lefteqn{\mathbb{E}\left[e^{-\int_0^{\tau_{l, u}}\eta(X_w)dw} (e^{X_{\tau_{l, u}}} - K)^{+}\mid X_0=x\right]}\nonumber\\
&&= \mathbb{E}^{(1)}\left[e^{-\int_0^{\tau_{\frac{s}{u}K, \frac{s}{l}K}}(\omega(\frac{1}{\hat{S}_w} sK) -\psi(1)) dw}(s - e^{\hat{X}_{\tau_{\frac{s}{u}K, \frac{s}{l}K}}})^{+}\mid \hat{X}_0=\log K\right],\label{changecallput}
\end{eqnarray}
where $\hat{S}_t = e^{\hat{X}_t}$ and $\hat{X}_t = -X_t$ is the dual process to $X_t$ and from \cite{tumilewicz, mordecki, exponentialMartingale} it follows that under $\mathbb{P}^{(1)}$ it is again L\'evy process with the triple
$(-\zeta, \sigma, \hat{\Pi})$ for $\hat{\Pi}$ defined in \eqref{hatPi}.
This completes the proof of identity \eqref{pierwszatozsamosc}.
Now, we recall that from general stopping theory we know that optimal stopping region for the call option is
of the form $\tau_c=\inf\{t \geq 0: S_t \in D_c\}$ for some stopping set $D_c$; see e.g. \cite[Thm. 2.4, p. 37]{peskir}.
Performing the same transformation like it is done in \eqref{changecallput} with $\tau_c$ instead of $\tau_{l, u}$
we can conclude that the optimal stopping time for the call option is the same as
the stopping time $\tau^*=\inf\{t \geq 0: \hat{S}_t \in D\}$ for the put option (replacing
$\tau_{\frac{s}{u}K, \frac{s}{l}K}$ in this transformation on the right-hand side),
where $D:=\{K\frac{s}{x} \quad\mbox{and}\quad x\in D_c\}$.
But from Theorem \ref{lemmaAboutOptimalStoppingRule} it follows that the optimal stopping time for the put option is the first entrance time
to some optimal interval. Thus from above considerations it follows that the stopping region
for the call option is of the same type as for the put option
and therefore \eqref{symm} follows from \eqref{pierwszatozsamosc}.
\end{proof}

\begin{proof}[\textbf{Proof of Theorem \ref{mainTheoremBSModel}}]
We prove that for the function $h$ satisfying \eqref{ODEForh}
we have
\begin{equation}\label{tobeproved}
\mathbb{E}_s\left[\frac{h(S_{\tau_{l, u}})}{h(s)} e^{-\int_0^{\tau_{l, u}}\omega(S_w)dw}\right] = 1.
\end{equation}
Since process $S_t$ is continuous in Black-Scholes model, $S_{\tau_{l, u}}$ equals either to $l$ or $u$ depending on the initial state of $S_t$.
We can distinguish three possible scenarios
\begin{enumerate}
\item $s<l$:
As the process $S_t$ is a continuous process and starts below the interval
$[l, u]$, then $\tau_{l, u} = \tau^{+}_l$ and $S_{\tau_{l, u}} = l$. Thus, we get
\begin{equation}\label{taul+}
\begin{aligned}
v^{\omega}_{\text{\rm A}^{\text{\rm Put}}}(s, l, u) &= \mathbb{E}_{s}\left[e^{-\int_0^{\tau^{+}_l}\omega(S_w)dw}; S_{\tau^{+}_l} = l\right](K-l) \\ &= \frac{h(s)}{h(l)}(K-l).
\end{aligned}
\end{equation}
\item $s\in[l, u]$:
If the process $S_t$ starts inside the interval $[l, u]$ which is the optimal stopping region, we decide to exercise our option immediately, i.e.
$\tau_{l, u} = 0$. Therefore, we have
\begin{equation}\label{tauTo0}
v^{\omega}_{\text{\rm A}^{\text{\rm Put}}}(s, l, u) = K-s.
\end{equation}
\item $s>u$:
Similarly to the case when $s<l$, the process $S_t$ can enter $[l, u]$ only via $u$ and thus $\tau_{l, u} = \tau^{-}_u$ and $S_{\tau_{l, u}} = u$. Therefore,
\begin{equation}\label{tauu-}
\begin{aligned}
v^{\omega}_{\text{\rm A}^{\text{\rm Put}}}(s, l, u) &= \mathbb{E}_{s}\left[e^{-\int_0^{\tau^{-}_u}\omega(S_w)dw}; S_{\tau^{-}_u} = u\right](K-u) \\ &= \frac{h(s)}{h(u)}(K-u).
\end{aligned}
\end{equation}
\end{enumerate}
Identities \eqref{taul+}, \eqref{tauTo0} and \eqref{tauu-} give the first part of the assertion of the theorem.
Note that boundary condition \eqref{ConditionsOnh(x)2} follows straightforward from the definition of
the value function of the American put option.

We are left with the proof of \eqref{tobeproved}.
Consider strictly positive and bounded by some $C$ function $h \in C^2(\mathbb{R}^+)\subset D(\mathcal{A})$.
Then by \cite[Prop. 3.2]{exponentialMartingale} the process
\begin{equation*}\label{Eh(t)}
E^h(t) := \frac{h(S_t)}{h(S_0)}e^{-\int_0^t \frac{(\mathcal{A} h)(S_w)}{h(S_w)} dw},
\end{equation*}
is a mean-one local martingale, where in the case of Black-Scholes model
\begin{equation*}\label{infinitesimalGenOfS}
\mathcal{A} h(s) = \mu s h'(s) + \frac{\sigma^2 s^2}{2} h''(s).
\end{equation*}
Observe that equation \eqref{ODEForh} is equivalent to
\begin{equation*}\label{omegaFunction}
\omega(s) = \frac{\mathcal{A} h(s)}{h(s)}.
\end{equation*}
Let
\begin{equation*}
\tau^M_{l, u} := \tau_{l,u}\wedge M
\end{equation*}
for some fixed $M>0$.
Applying the optional stopping theorem for bounded stopping time, we derive
\begin{equation}\label{expectationForItoDiffusion}
\mathbb{E}_s\left[\frac{h(S_{\tau^M_{l, u}})}{h(s)} e^{-\int_0^{\tau^M_{l, u}}  \omega(S_w)dw}\right] = 1.
\end{equation}
We rewrite the left side of \eqref{expectationForItoDiffusion} as a sum of the following two components
\begin{align*}
I_1 &:= \mathbb{E}_s\left[\frac{h(S_{\tau^M_{l, u}})}{h(s)} e^{-\int_0^{\tau^M_{l, u}}  \omega(S_w)dw}; \tau_{l,u}>M\right], \\
I_2 &:= \mathbb{E}_s\left[\frac{h(S_{\tau^M_{l, u}})}{h(s)} e^{-\int_0^{\tau^M_{l, u}} \omega(S_w)dw}; \tau_{l,u}\leq M\right].
\end{align*}
We prove now that $\lim_{M\rightarrow\infty} I_1 = 0$ and $\lim_{M\rightarrow\infty} I_2 \in(0, +\infty)$.
Let us define the last time when value function \eqref{valueFunctionOmegaAmericanPutOption2} is positive by
\begin{equation*}
\tau_{\text{last}}(K):=\sup\{t\geq 0: S_t\leq K\}.
\end{equation*}
It easy to notice that $\mathbb{P}(\tau^M_{l,u}\leq \tau_{\text{last}}(K)) = 1$.
Then, from the boundedness of $h$, lower boundedness of $\omega$ and Cauchy–Schwarz inequality we obtain
\begin{align*}
I_1 &\leq \frac{C}{h(s)}\mathbb{E}\left[e^{-\ubar{\omega}\tau_{\text{last}}(K)}; \tau_{l,u}>M\right] \\
&= \frac{C}{h(s)}\mathbb{E}\left[e^{-\ubar{\omega}\tau_{\text{last}}(K)} \mathds{1}_{\tau_{l,u}>M}\right] \\
&\leq \frac{C}{h(s)} \sqrt{\mathbb{E}\left[e^{-2\ubar{\omega}\tau_{\text{last}}(K)}\right] \mathbb{P}(\tau_{l,u}>M)},
\end{align*}
where $\ubar{\omega} := \min_{s\in\mathbb{R}^{+}}\omega(s)$.
Note that $\sqrt{\mathbb{E}\left[e^{-2\ubar{\omega}\tau_{\text{last}}(K)}\right]}<\infty$ by \cite[Thm. 2]{eriklast}
because $\mathbb{E}e^{-2\ubar{\omega} B_t}<\infty$ for any $t\geq 0$.
Thus $\lim_{M\rightarrow\infty} I_1 = 0$.
Moreover,
\begin{align*}
0 < I_2 \leq \frac{C}{h(s)}\mathbb{E}\left[e^{-\ubar{\omega}\tau_{\text{last}}(K)}; \tau_{l,u}<M\right].
\end{align*}
Hence by \eqref{expectationForItoDiffusion} and the dominated convergence we get \eqref{tobeproved} as long as $h$ is positive and bounded.
Finally, since $S_{\tau_{l, u}}$ equals either to $l$ or $u$, the boundedness assumption could be skipped.
This completes the proof.
\end{proof}

\begin{proof}[\textbf{Proof of Theorem \ref{mainTheoremBSModelExpojump0}}]
From Theorem \ref{mainTheoremAmericanOptionConvexity} and Remark \ref{pierwszauwaga}
it follows that the optimal exercise time is the first entrance to the
interval $[l^*, u^*]$ and by Theorem \ref{lemmaAboutOptimalStoppingRule} the value function $V^{\omega}_{\text{\rm A}^{\text{\rm Put}}}(s)$ is equal to the maximum over $l$ and $u$ of
$v^{\omega}_{\text{\rm A}^{\text{\rm Put}}}(s, l, u)$ defined in \eqref{valueFunctionOmegaAmericanPutOption2}.
We recall the observation that if the discount function
$\omega$ is nonnegative, then it is never
optimal to wait to exercise option for small asset prices, that is, always $l^*=0$ in this case
and the stopping region is one-sided.
We find now function $v^{\omega}_{\text{\rm A}^{\text{\rm Put}}}(s, l, u)$ in the case of (i) and (ii).

If $\sigma=0$, by the lack of memory of exponential random variable, using similar analysis like in the proof of Theorem \ref{mainTheoremGeometricLevy}, we have
\begin{align*}v^{\omega}_{\text{\rm A}^{\text{\rm Put}}}(s, 0, u)&=\mathbb{E}(K-e^{\log u- Y})^+ \mathbb{E}_{s}\left[e^{-\int_0^{\tau_{u}^-}\omega(S_w)dw}; \tau^{-}_{u}<\infty\right]\\
&=
\mathbb{E}(K-e^{\log u-Y})^+\mathbb{E}\left[e^{-\int_0^{\sigma_{0}^-}\eta_u(X_w)dw}; \sigma^{-}_{0}<\infty\mid X_0=x-\log u\right]\\
&=
\mathbb{E}(K-e^{\log u-Y})^+\left(\mathcal{Z}^{(\eta_u)}(x-\log u)- c_{\mathcal{Z}^{(\eta_u)}/\mathcal{W}^{(\eta_u)}} \mathcal{W}^{(\eta_u)}(x-\log u)\right)\\
&=
\mathbb{E}(K-e^{\log u-Y})^+\left(\mathscr{Z}^{(\omega_u)}\left(\frac{s}{u}\right)- c_{\mathscr{Z}^{(\omega_u)}/\mathscr{W}^{(\omega_u)}} \mathscr{W}^{(\omega_u)}\left(\frac{s}{u}\right)\right).
\end{align*}
Observing that
\[\mathbb{E}(K-e^{\log u-Y})^+= K-\frac{u \varphi}{\varphi+1}\]
completes the proof of part (i).

If $\sigma>0$ then
\begin{align*}
v^{\omega}_{\text{\rm A}^{\text{\rm Put}}}(s, 0, u)
&=
\mathbb{E}(K-e^{\log u-Y})^+\mathbb{E}\left[e^{-\int_0^{\sigma_{0}^-}\eta_u(X_w)dw}; \sigma^{-}_{0}<\infty, X_{\sigma^{-}_{0}}<0\mid X_0=x-\log u\right]\\
&\qquad+(K-u)\mathbb{E}\left[e^{-\int_0^{\sigma_{0}^-}\eta_u(X_w)dw}; \sigma^{-}_{0}<\infty, X_{\sigma^{-}_{0}}=0\mid X_0=x-\log u\right].
\end{align*}
The first increment can be analysed like in the case of $\sigma=0$. The expression for the second component
follows from \eqref{formulaWithLimit2}.

Finally, the smooth fit condition follows straightforward from Theorem \ref{smoothTheorem}.
\end{proof}

\begin{proof}[\textbf{Proof of Theorem \ref{mainTheoremBSModelExpojump}}]
Assume first that $\sigma=0$. Then
\begin{equation}\label{W(x)}
W(x)= \Upsilon_1 e^{\gamma_1 x} + \Upsilon_2 e^{\gamma_2 x}
\end{equation}
with $\gamma_1=0$.
To produce ordinary differential equation for $\mathcal{W}^{(\xi)}(x)$ we start from equation
\eqref{scaleFunctionW}. Putting \eqref{W(x)} there gives
\begin{equation}\label{W}
\mathcal{W}^{(\xi)}(x) = \Upsilon_1 + \Upsilon_2 e^{\gamma_2 x} + \Upsilon_1\int_0^{x}\xi(y)\mathcal{W}^{(\xi)}(y)dy + \Upsilon_2 \int_0^{x}e^{\gamma_2 (x-y)}\xi(y)\mathcal{W}^{(\xi)}(y)dy.
\end{equation}
Taking the derivative of both sides gives
\begin{equation}\label{deriv_W}
{\mathcal{W}^{(\xi)}}'(x) = \Upsilon_2\gamma_2 e^{\gamma_2 x} + \Upsilon_1 \xi(x) \mathcal{W}^{(\xi)}(x) + \Upsilon_2\left(\xi(x)\mathcal{W}^{(\xi)}(x) + \gamma_2\int_0^{x} e^{\gamma_2 (x-y)} \xi(y)\mathcal{W}^{(\xi)}(y)dy\right).
\end{equation}
From \eqref{W} we have
\begin{equation*}
\int_0^{x} e^{\gamma_2 (x-y)}\xi(y)\mathcal{W}^{(\xi)}(y)dy = \frac{1}{\Upsilon_2} \left(\mathcal{W}^{(\xi)}(x) - \Upsilon_1 - \Upsilon_2 e^{\gamma_2 x} - \Upsilon_1\int_0^{x}\xi(y)\mathcal{W}^{(\xi)}(y)dy \right).
\end{equation*}
We put it into \eqref{deriv_W} and derive
%
\begin{equation*}
{\mathcal{W}^{(\xi)}}'(x) = ((\Upsilon_1+\Upsilon_2)\xi(x)+\gamma_2)\mathcal{W}^{(\xi)}(x) - \gamma_2 \Upsilon_1 - \gamma_2 \Upsilon_1\int_0^{x}\xi(y)\mathcal{W}^{(\xi)}(y)dy.
\end{equation*}
We take the derivative of both sides again to get equation \eqref{scaleequation}.

From \eqref{scaleFunctionW}, \eqref{W(x)} and \eqref{deriv_W} we derive both initial conditions \eqref{fbW}.

Similar analysis can be done for the $\mathcal{Z}^{(\xi)}(x)$ scale function producing equation \eqref{scaleequation} and its initial conditions.
This completes the proof of the case (i).

In the case when $\sigma>0$ observe that
\begin{equation}\label{W(x)2}
W(x)= \Upsilon_1 e^{\gamma_1 x} + \Upsilon_2 e^{\gamma_2 x} + \Upsilon_3 e^{\gamma_3 x}
\end{equation}
with $\gamma_1 = 0$. Thus, from \eqref{scaleFunctionW} $\mathcal{W}^{(\xi)}(x)$ satisfies the following equation
\begin{equation*}
\mathcal{W}^{(\xi)}(x) = \Upsilon_1 + \Upsilon_2 e^{\gamma_2 x} + \Upsilon_3 e^{\gamma_3 x} + \int_0^{x}(\Upsilon_1 + \Upsilon_2 e^{\gamma_2 (x-y)} + \Upsilon_3 e^{\gamma_3 (x-y)})\xi(y)\mathcal{W}^{(\xi)}(y)dy.
\end{equation*}
We simplify it deriving
\begin{equation}\label{W2}
\begin{split}
\mathcal{W}^{(\xi)}(x) &= \Upsilon_1 + \Upsilon_2 e^{\gamma_2 x} + \Upsilon_3 e^{\gamma_3 x} + \Upsilon_1\int_0^{x}\xi(y)\mathcal{W}^{(\xi)}(y)dy \\ &+ \Upsilon_2\int_0^{x} e^{\gamma_2 (x-y)}\xi(y)\mathcal{W}^{(\xi)}(y)dy  + \Upsilon_3\int_0^{x} e^{\gamma_3 (x-y)}\xi(y)\mathcal{W}^{(\xi)}(y)dy.
\end{split}
\end{equation}
In the next step we take derivative of both sides to get
\begin{equation}\label{deriv_W2}
\begin{split}
{\mathcal{W}^{(\xi)}}'(x) &= \Upsilon_2\gamma_2 e^{\gamma_2 x} + \Upsilon_3\gamma_3 e^{\gamma_3 x}+ \Upsilon_1 \xi(x) \mathcal{W}^{(\xi)}(x) + \Upsilon_2\left(\xi(x)\mathcal{W}^{(\xi)}(x) + \gamma_2\int_0^{x} e^{\gamma_2 (x-y)} \xi(y)\mathcal{W}^{(\xi)}(y)dy\right) \\ +
&\Upsilon_3\left(\xi(x)\mathcal{W}^{(\xi)}(x) + \gamma_3\int_0^{x} e^{\gamma_3 (x-y)} \xi(y)\mathcal{W}^{(\xi)}(y)dy\right).
\end{split}
\end{equation}
From \eqref{W2} we have
\begin{equation*}
\begin{split}
\int_0^{x} e^{\gamma_3 (x-y)}\xi(y)\mathcal{W}^{(\xi)}(y)dy = \frac{1}{\Upsilon_3} \left(\mathcal{W}^{(\xi)}(x) - \Upsilon_1 - \Upsilon_2 e^{\gamma_2 x} - \Upsilon_3 e^{\gamma_3 x} \right. \\ \left. - \Upsilon_1\int_0^{x}\xi(y)\mathcal{W}^{(\xi)}(y)dy - \Upsilon_2 \int_0^{x} e^{\gamma_2 (x-y)}\xi(y)\mathcal{W}^{(\xi)}(y)dy \right).
\end{split}
\end{equation*}
We put it into \eqref{deriv_W2} deriving
\begin{equation}\label{deriv_simpl_W2}
\begin{split}
{\mathcal{W}^{(\xi)}}'(x) &= \Upsilon_2(\gamma_2-\gamma_3) e^{\gamma_2 x} + (\Upsilon_1 + \Upsilon_2 + \Upsilon_3)\xi(x)\mathcal{W}^{(\xi)}(x) \\ &+ \Upsilon_2 (\gamma_2-\gamma_3) \int_0^{x} e^{\gamma_2 (x-y)} \xi(y)\mathcal{W}^{(\xi)}(y)dy + \gamma_3\mathcal{W}^{(\xi)}(x) - \gamma_3 \Upsilon_1 \\&- \gamma_3 \Upsilon_1 \int_0^{x}\xi(y)\mathcal{W}^{(\xi)}(y)dy.
\end{split}
\end{equation}

Taking again derivative of both sides produces
\begin{equation}\label{deriv2_simpl_W2}
\begin{split}
{\mathcal{W}^{(\xi)}}''(x) &= \Upsilon_2(\gamma_2-\gamma_3)\gamma_2 e^{\gamma_2 x} + (\Upsilon_1 + \Upsilon_2 + \Upsilon_3) (\xi'(x)\mathcal{W}^{(\xi)}(x) + \xi(x){\mathcal{W}^{(\xi)}}'(x)) \\ &+ \Upsilon_2(\gamma_2-\gamma_3) \left(\xi(x)\mathcal{W}^{(\xi)}(x) + \gamma_2\int_0^{x} e^{\gamma_2 (x-y)} \xi(y)\mathcal{W}^{(\xi)}(y)dy\right) + \gamma_3{\mathcal{W}^{(\xi)}}'(x) \\ &- \gamma_3 \Upsilon_1 \xi(x)\mathcal{W}^{(\xi)}(x).
\end{split}
\end{equation}
From \eqref{deriv_simpl_W2} we have
\begin{equation*}
\begin{split}
\int_0^{x} e^{\gamma_2 (x-y)} \xi(y)\mathcal{W}^{(\xi)}(y)dy &= \frac{1}{\Upsilon_2 (\gamma_2-\gamma_3)} \left({\mathcal{W}^{(\xi)}}'(x) - \Upsilon_2(\gamma_2-\gamma_3) e^{\gamma_2 x} \right. \\&- \left. (\Upsilon_1 + \Upsilon_2 + \Upsilon_3)\xi(x)\mathcal{W}^{(\xi)}(x)- \gamma_3\mathcal{W}^{(\xi)}(x) + \gamma_3 \Upsilon_1 \right. \\& \left. + \gamma_3 \Upsilon_1 \int_0^{x}\xi(y)\mathcal{W}^{(\xi)}(y)dy\right).
\end{split}
\end{equation*}
We put it into \eqref{deriv2_simpl_W2} to get
\begin{equation*}
\begin{split}
{\mathcal{W}^{(\xi)}}''(x) &= (\Upsilon_1 + \Upsilon_2 + \Upsilon_3) (\xi'(x)\mathcal{W}^{(\xi)}(x) + \xi(x){\mathcal{W}^{(\xi)}}'(x)) \\ &+ \Upsilon_2(\gamma_2-\gamma_3) \xi(x)\mathcal{W}^{(\xi)}(x) \\ &+ \gamma_2\left({\mathcal{W}^{(\xi)}}'(x) - (\Upsilon_1+\Upsilon_2+\Upsilon_3) \xi(x) \mathcal{W}^{(\xi)}(x) - \gamma_3 \mathcal{W}^{(\xi)}(x) \right. \\ &\left. + \gamma_3 \Upsilon_1 + \gamma_3 \Upsilon_1 \int_0^{x}\xi(y)\mathcal{W}^{(\xi)}(y)dy \right) \\ &+ \gamma_3{\mathcal{W}^{(\xi)}}'(x) - \gamma_3 \Upsilon_1 \xi(x)\mathcal{W}^{(\xi)}(x).
\end{split}
\end{equation*}
Taking again derivative and simplifying gives
%
\begin{equation*}
\begin{split}
{\mathcal{W}^{(\xi)}}'''(x) &= \left((\Upsilon_1+\Upsilon_2+\Upsilon_3)\xi(x)+\gamma_2+\gamma_3\right){\mathcal{W}^{(\xi)}}''(x) \\&+ \left(2(\Upsilon_1+\Upsilon_2+\Upsilon_3)\xi'(x)+\Upsilon_2(\gamma_2-\gamma_3)\xi(x)-(\Upsilon_1+\Upsilon_2+\Upsilon_3)\gamma_2\xi(x)-\gamma_2\gamma_3-\gamma_3\Upsilon_1\xi(x)\right){\mathcal{W}^{(\xi)}}'(x) \\&+ \left((\Upsilon_1+\Upsilon_2+\Upsilon_3)\xi''(x)+\Upsilon_2(\gamma_2-\gamma_3)\xi'(x)-\gamma_2(\Upsilon_1+\Upsilon_2+\Upsilon_3)\xi'(x)+\gamma_2\gamma_3 \Upsilon_1 \xi(x)-\gamma_3 \Upsilon_1\xi'(x)\right){\mathcal{W}^{(\xi)}}(x).
\end{split}
\end{equation*}
Ultimately, taking into account \eqref{scaleFunctionW} and knowing that $W(x)=0$ we conclude that $\Upsilon_1+\Upsilon_2+\Upsilon_3 = 0$. Hence, we obtain the equation that we wanted to prove.

From \eqref{W(x)2} and \eqref{scaleFunctionW} we have
\begin{equation*}
{\mathcal{W}^{(\xi)}}(0) = \Upsilon_1 + \Upsilon_2 + \Upsilon_3 = 0.
\end{equation*}
From \eqref{deriv_simpl_W2} it follows that
\begin{equation*}
{\mathcal{W}^{(\xi)}}'(0) = \Upsilon_2\gamma_2 + \Upsilon_3\gamma_3 + (\Upsilon_1 + \Upsilon_2 + \Upsilon_3)^2 \xi(0) = \Upsilon_2\gamma_2 + \Upsilon_3\gamma_3.
\end{equation*}
Finally, from \eqref{deriv2_simpl_W2} we have
\begin{equation*}
\begin{split}
{\mathcal{W}^{(\xi)}}''(0) &= \Upsilon_2 \gamma_2 (\gamma_2-\gamma_3) + (\Upsilon_1 + \Upsilon_2 + \Upsilon_3) (\xi'(0) {\mathcal{W}^{(\xi)}}(0) + \xi(0) {\mathcal{W}^{(\xi)}}'(0)) \\ &+ \Upsilon_2(\gamma_2-\gamma_3) \xi(0) {\mathcal{W}^{(\xi)}}(0) + \gamma_3 {\mathcal{W}^{(\xi)}}'(0) - \gamma_3 \Upsilon_1 \xi(0) {\mathcal{W}^{(\xi)}}(0) \\
&= \Upsilon_2{\gamma_2}^2 + \Upsilon_3{\gamma_3}^2.
\end{split}
\end{equation*}
The analysis for $\mathcal{Z}^{(\xi)}(x)$ can be done in the same way.
This completes the proof.
\end{proof}

\section{Appendix}

\begin{proof}[\textbf{Proof of Lemma \ref{TheoremAboutExistenceOfValueFunctionDiffusionWithJumps}}]

Firstly, we define the function $f:\mathbb{R}^{+}\rightarrow\mathbb{R}$ of the form
\begin{equation*}
f(s) =
\begin{cases}
-\frac{1}{s}, \quad &s\in(0, 1],\\
s, \quad &s\in[2, +\infty)
\end{cases}
\end{equation*}
such that $f(s)\in C^{2}(\mathbb{R}^{+})$ and $f'(s)>0$ for all $s\in\mathbb{R}^{+}$.

Taking $Y_t = f(S_t)$ and applying It\^{o}'s lemma on \eqref{diffusionWithJumps}, we obtain
\begin{equation*}
dY_t = \tilde{\mu}(Y_{t-}, t) dt + \tilde{\sigma}(Y_{t-}, t)dB_t + \int_{\mathbb{R}} \tilde{\gamma}(Y_{t-}, t, z)\tilde{v}(dt, dz),
\end{equation*}

where
\begin{align*}
  \phantom{i + j + k}
  &\begin{aligned}
    &\mathllap{\tilde{\mu}(y, t)} = \mu(f^{-1}(y), t) f'(f^{-1}(y)) + \frac{\sigma^2(f^{-1}(y), t)}{2} f''(f^{-1}(y))\\
      &\qquad + \int_{\mathbb{R}} \left(\tilde{\gamma}(y, t, z) - f'(f^{-1}(y))\gamma(f^{-1}(y), t, z)\right) m(dz),
  \end{aligned}\\
  &\begin{aligned}
    &\mathllap{\tilde{\sigma}(y, t)} = f'(f^{-1}(y))\sigma(f^{-1}(y), t),
  \end{aligned}\\
  &\begin{aligned}
    &\mathllap{\tilde{\gamma}(y, t, z)} = f(f^{-1}(y) + \gamma(f^{-1}(y), t, z)) - y.
  \end{aligned}
\end{align*}
We define also the function
\begin{equation*}
\tilde{\omega}(y) := \omega(f^{-1}(y))
\end{equation*}
and
\begin{equation*}
\tilde{g}(y) := g(f^{-1}(y)).
\end{equation*}
We can now verify that the functions $\tilde{\mu}(y, t)$, $\tilde{\sigma}(y, t)$, $\tilde{\gamma}(y, t, z)$ and $\tilde{g}(y)$ satisfy conditions $(2.2)-(2.5)$ from \cite[Section 2]{pham98}.
Let
\begin{equation*}
v(y, t): = V^{\omega}_{\text{\rm E}}(f^{-1}(y), t).
\end{equation*}
From \cite[Theorem 3.1]{pham98} it follows that $v(y,t)$ is a viscosity solution to
\begin{equation}\label{CauchyProblem2}
  \begin{cases}
    \tilde{\mathcal{L}} v(y, t) = \tilde{f}(y,t), \quad & (y, t)\in\mathbb{R}\times[0,T), \\
	v(y, T) = \tilde{g}(y), \quad & y\in\mathbb{R},
  \end{cases}
\end{equation}
where
\begin{equation*}
\tilde{\mathcal{L}} v(y, t) = - \frac{\partial v(y, t)}{\partial t} - \frac{\tilde{\sigma}^2(y, t)}{2} \frac{\partial^2 v(y,t)}{\partial y^2} - \hat{\mu}(y,t) \frac{\partial v(y, t)}{\partial y} + \tilde{\omega}(y)v(y,t)
\end{equation*}
with \begin{equation*}
\hat{\mu}(y,t) = \tilde{\mu}(y,t) - \int_{\mathbb{R}}\tilde{\gamma}(y,t,z)m(dz)
\end{equation*}
and
\begin{equation*}
\tilde{f}(y,t) = - \int_{\mathbb{R}}\left(v(y+\tilde{\gamma}(y,t,z), t) - v(y,t)\right)m(dz).
\end{equation*}

In addition, using \cite[Prop. 3.3]{pham98} yields that $v(y,t)\in C(\mathbb{R}\times[0,T])$ and it satisfies
\begin{equation}\label{inequalityOnV}
|v(y_2,t_2) - v(y_1,t_1)|\leq C((1 + |y_2|)|t_2 - t_1|^{\frac{1}{2}} + |y_2 - y_1|)
\end{equation}
for some $C>0$ and for all $t_1, t_2 \in[0,T]$ and $y_1, y_2\in \mathbb{R}$.
Based on \eqref{inequalityOnV} and assumptions put on $\gamma$ we can conclude that $\tilde{f}(y,t)\in C_{\alpha}(\mathbb{R}\times[0,T])\cap C_{\text{pol}}(\mathbb{R}\times[0,T])$.
Then applying \cite[Thm. A.14]{tysk} give us the existence of a unique classical solution $w(y,t)$ to \eqref{CauchyProblem2} such that $w(y,t)\in C^{2,1}(\mathbb{R}\times [0, T)) \cap C_{\text{pol}}(\mathbb{R}\times[0,T])$.
In view of the fact that $w(y,t)$ is continuous, we can observe that $\tilde{f}(y,t)$ is Lipschitz continuous in $y$, uniformly in $t$.
Hence from \cite[Lem. 3.1]{pham98} we know that that $w(y,t)$ is also Lipschitz continuous in $y$, uniformly in $t$.
From the uniqueness result given in \cite[Thm. 4.1]{pham98} we can deduce that $v(y,t) = w(y,t)$.
Applying \cite[Thm. A.18.]{tysk} we find that $v(y,t)\in C^{4,1}_{\alpha}(\mathbb{R}\times[0,T])$.
Changing back to the original coordinates, it follows that $V^{\omega}_{\text{\rm E}}(s,t)\in C^{4,1}_{\alpha}(\mathbb{R}^{+}\times[0,T])\cap C_{\text{pol}}(\mathbb{R}^{+}\times[0,T])$ and it satisfies \eqref{CauchyProblem}.
\end{proof}

\begin{proof}[\textbf{Proof of Lemma \ref{LemmaAboutBoundednessOfSecondDerivativeDiffusionWithJumps}}]

The proof follows in the same way as the proof of Lemma \ref{TheoremAboutExistenceOfValueFunctionDiffusionWithJumps}. However, this time we apply \cite[Thm. A.20]{tysk} which guarantees the existence of a unique classical solution $w(y,t)$ to \eqref{CauchyProblem2} satisfying $w(y, t) \in C^{2,1}_{\text{pol}}(\mathbb{R}\times[0,T])$. Hence, coming back to the original coordinates, we have that $V^{\omega}_{\text{\rm E}}(s, t) \in C^{2,1}_{\text{pol}}(\mathbb{R}^{+}\times[0,T])$. Therefore, there exist constants $n>0$ and $K>0$ such that
\begin{equation*}
\left|\frac{\partial^2 V^{\omega}_{\text{\rm E}}(s, t)}{\partial s^2}\right| \leq K(s^{-n}+s^{n})
\end{equation*}
for all $(s, t)\in\mathbb{R}^{+}\times[0,T]$. This completes the proof.
\end{proof}

\begin{proof}[\textbf{Proof of Lemma \ref{TheoremAboutConvexityOfBermudan}}]
By the dynamic programming principle formulated e.g. in \cite{propertiesOfAmericanOption}, the value function $V^{\omega}_{\text{B}_\Xi}(s, t)$ satsifies
\begin{enumerate}
\item At time $t=T$, the value function $V^{\omega}_{\text{B}_\Xi}(s, t)$ is equal to $g(s)$.
\item Given the price $V^{\omega}_{\text{B}_\Xi}(s, t_n)$ at the time
$t_n = \frac{n}{2^\Xi}T$, the price at time $t_{n-1} = \frac{n-1}{2^\Xi}T $ is $V^{\omega}_{\text{B}_\Xi}(s, t_{n-1}) = \max\{\mathbb{E}_{s, t_{n-1}}
[e^{-\int_{t_{n-1}}^{t_n} \omega(S_w)dw} V^{\omega}_{\text{B}_\Xi}(S_{t_n}, t_n)], g(s)\}$.
\end{enumerate}

Thus, the price $V^{\omega}_{\text{B}_\Xi}(s, t_{n-1})$ of a Bermudan option at $t=t_{n-1}$ can be calculated inductively as the maximum of the payoff function $g$ and the price of a European option with expiry $t_n$ and payoff function $V^{\omega}_{\text{B}_\Xi}(S_{t_n}, t_n)$. From Theorem \ref{TheoremAboutConvexityOfValueFunctionDiffusionWithJumps} we know that the value function of European option is convex in $s$ provided the payoff function is convex, and since the maximum of two convex functions is again a convex function, we conclude that the Bermudan option price $V^{\omega}_{\text{B}_\Xi}(s, t)$
is convex in $s$ for all $t\in[0,T]$.

\end{proof}

\section{Concluding remarks}

In this paper, we have identified the value function in the optimal stopping problem with asset-dependent
discounting. We have also shown that under certain conditions we are able to obtain a closed form of the value function.
We have provided some specific examples as well.

It is tempting to analyse other discount functions. For example $\omega$ might ba a random function
or just simply a random variable dependent on the asset process $S_t$. One can take other processes as a discount rate
where the dependence is introduced not only via correlation between gaussian components but via common jump structure.
This jump-dependence is crucial since crashes in the market affect large portion of business at the same time; see e.g. \cite{reno}.

One can take a Poisson version of American options where exercise is possible only at independent Poisson epochs as well.
First attempt for classical perpetual American options has been already made in \cite{meandothers}.
We believe that present analysis can be generalised to this set-up.

Obviously, it would be good to work out details for different payoff functions, hence for various options.
One could think of barrier options, Russian, Israeli or Swing options.
What is maybe even more interesting for the future analysis is taking into account Markov switching markets
and using omega scale matrices introduced in \cite{meandothers2}. We expect that in this setting the optimal exercise time is also the first entrance time to the interval which ends depend on the governing Markov chain.


\end{document}